\documentclass[sigconf,nonacm=true,anonymous=false,natbib=false,balance=true,pbalance=false]{acmart}

\subtitle{Full Version}

\newcommand{\fullversionnote}{%
	\authornote{This is an extended version of~\cite{morio2026}.}%
}

\usepackage{subfiles}

\newif\iffullversion
\NewDocumentEnvironment{full}{+b}{%
	\iffullversion#1\fi
}{}
\NewDocumentEnvironment{conf}{+b}{%
	\iffullversion\else#1\fi
}{}

\newcommand{\fullversionref}{msr2pv-full}
\newcommand{\fullversionurl}{https://arxiv.org/pdf/2608.06315}
\newcommand{\fullversionlocation}[1]{%
	\href{\fullversionurl\#nameddest=\getrefbykeydefault{F-#1}{anchor}{}}{\Cref*{F-#1}}%
}
\newcommand{\fullorappendix}[1]{%
	\iffullversion
		\cref{#1}%
	\else
		the full version~\cite[\fullversionlocation{#1}]{\fullversionref}%
	\fi
}

\usepackage[english]{babel}
\usepackage[autostyle]{csquotes}
\usepackage{microtype}

\usepackage{subcaption}
\usepackage{amsthm}
\usepackage{thmtools,thm-restate} 
\usepackage{amsmath}
\usepackage{hyperref}
\usepackage{xr}
\usepackage[capitalize, noabbrev]{cleveref}
\usepackage{proof} %
\usepackage{syntax}
\usepackage[inline]{enumitem}
\usepackage{mathtools} %
\usepackage{tabularx}
\usepackage{booktabs}
\usepackage{multirow}
\usepackage{siunitx}
\usepackage{placeins}

\usepackage{tikz}
\usetikzlibrary{arrows.meta, positioning, quotes, babel}

\usepackage{listings}
\AtBeginDocument{%
  \lstset{
    numbers=none,
    frame=none,
    xleftmargin=1.5em
  }%
}

\usepackage{packages/definitions} %
\usepackage{packages/macros-basic}
\usepackage{packages/macros-dy-theory}
\usepackage{packages/macros-msr}
\usepackage{packages/macros-proverif}
\usepackage{packages/listings-proverif} %

\theoremstyle{definition}
\newtheorem{definition}{Definition}
\newtheorem{remark}{Remark}
\newtheorem{example}{Example}
\newtheorem{lemma}{Lemma}
\newtheorem{corollary}{Corollary}

\declaretheorem{theorem} %

\usepackage[
	backend=biber,
	natbib=true,
	datamodel=acmdatamodel,
	style=acmnumeric,
	maxcitenames=2,
	mincitenames=1,
	maxbibnames=999,
	url=false,
	doi=false,
	isbn=false,
]{biblatex}

\DeclareFieldFormat{titlecase}{#1}

\DeclareSourcemap{
	\maps[datatype=bibtex]{
		\map{
			\step[fieldset=editor, null]
			\step[fieldset=month, null]
			\step[fieldset=address, null]
			\step[fieldset=location, null]
			\step[fieldset=issn, null]
			\step[fieldset=eventtitle, null]
			\step[fieldset=urldate, null]
			\step[fieldset=copyright, null]
			\step[fieldset=langid, null]
			\step[fieldset=keywords, null]
			\step[fieldset=collection, null]
		}
		\map{
			\pertype{inproceedings}
			\pertype{incollection}
			\step[fieldset=series, null]
			\step[fieldset=volume, null]
			\step[fieldset=number, null]
		}
		\map{
			\step[fieldsource=booktitle, match=\regexp{[.,]\s+(Proceedings|Revised\s+Selected\s+Papers)(,\s+Part\s+\{?[IVX]+\}?)?\s*$}, replace={}]
			\step[fieldsource=booktitle, match=\regexp{,\s+Proceedings,}, replace={,}]
			\step[fieldsource=booktitle, match=\regexp{(\s\d{4}),\s.+$}, replace={$1}]
			\step[fieldsource=booktitle, match=\regexp{,\s+\d{1,2}[-\x{2013}]\d{1,2}\s+[A-Z][a-z]+\s+\d{4}\s*$}, replace={}]
			\step[fieldsource=booktitle, match=\regexp{,\s+[A-Z][a-z]+\s+\d{1,2}\s*[-\x{2013}]\s*\d{1,2},\s+\d{4}\s*$}, replace={}]
			\step[fieldsource=booktitle, match=\regexp{,\s+[A-Z][A-Za-z\s.]*,\s+(UK|USA|United\s+Kingdom|Canada|Germany|France|Sweden)\s*$}, replace={}]
		}
	}
}

\AtBeginBibliography{\setlength{\emergencystretch}{1em}}

\usepackage{orcidlink}

\graphicspath{{figures/}}

\title[A Sound Translation from Tamarin to ProVerif: Enabling Comparative Analysis]{A Sound Translation from Tamarin to ProVerif:\texorpdfstring{\\}{ } Enabling Comparative Analysis}

\providecommand{\fullversionnote}{}

\author{Kevin Morio}
\orcid{0000-0002-0220-3448}
\email{kevin.morio@cispa.de}
\affiliation{%
	\institution{CISPA Helmholtz Center for Information Security}
	\city{Saarbrücken}
	\country{Germany}
}
\fullversionnote

\author{Yavor Ivanov}
\orcid{0009-0003-0815-6421}
\email{yavorivanov.yi@gmail.com}
\affiliation{%
  \institution{CISPA Helmholtz Center for Information Security}
  \city{Saarbrücken}
  \country{Germany}
}

\author{Robert Künnemann}
\orcid{0000-0003-0822-9283}
\email{robert.kuennemann@cispa.de}
\affiliation{%
	\institution{CISPA Helmholtz Center for Information Security}
	\city{Saarbrücken}
	\country{Germany}
}

\fullversiontrue
\begin{document}
\begin{abstract}
  Protocol verification tools enable the formal modeling and automatic verification of security protocols.
Two prominent tools in this area are Tamarin and ProVerif.
While they share the same high-level goal, they differ significantly in their underlying formalisms and verification techniques, making a systematic comparison challenging.
Tamarin uses multiset rewrite rules with sound and complete verification, whereas ProVerif employs an extension of the applied\nobreakdash-$\pi$ calculus that provides fast but potentially incomplete results.

In this work, we present a sound translation from Tamarin to ProVerif that enables a rigorous comparison of the two tools.
Our translation introduces techniques for formula rewriting, encoding multiset rewrite semantics, and handling simultaneous events, showing how Tamarin's features can be expressed in ProVerif's formalism while precisely characterizing the cases where this is not possible.
The translation supports a large subset of Tamarin's features---including multiset rewrite rules, lemmas and restrictions---thereby aligning the semantics of the two formalisms.
We compare the expressiveness of Tamarin's logic fragment with that of ProVerif, identifying which properties can be faithfully translated.
Moreover, we provide formal soundness and completeness proofs: within the faithful translation fragment, soundness ensures that any property verified in ProVerif also holds in the original Tamarin model, while completeness ensures that exists-trace properties not involving attacker knowledge are preserved by the translation.
Best-effort encodings, in particular XOR, are reported separately and are outside these guarantees.

Finally, we conduct an extensive evaluation of our translation on 121 Tamarin models.
The ProVerif front end accepts executable translations for 523 of 566 lemma tasks.
Among non-XOR tasks with definitive results from both tools, 237 of 238 agree, with the remaining verdict explicitly flagged as using an incomplete model.
Among the 344 tasks for which Tamarin returns a Boolean result and ProVerif completes with a logical result, ProVerif is faster in 316 cases (\qty{91.9}{\percent}), with median per-task runtime and peak-memory ratios of $\num{6.74}\times$ and $\num{6.13}\times$, respectively.

\end{abstract}

\maketitle

\hyphenation{pro-verif}

\section{Introduction}
\label{sec:introduction}

The design and verification of security protocols are challenging.
The complexity of protocol interactions and the subtlety of attacks make manual analysis error-prone, motivating the development of automated verification tools.
Such tools enable rigorous checking of whether protocols satisfy their intended security guarantees by expressing both the protocol and desired properties in a formal language, then employing symbolic methods to either discover attacks or establish their absence.
Two prominent tools in this domain are Tamarin~\cite{meier2013} and ProVerif~\cite{blanchet2001}.

While both tools pursue the same high-level objective, their divergent input languages and verification methodologies make it difficult to combine their complementary strengths.
Tamarin models protocols using multiset rewrite rules (MSRs), whereas ProVerif employs the applied\nobreakdash-$\pi$ calculus.
ProVerif is typically faster due to an abstraction step that avoids backtracking, though this may result in spurious attacks.
Tamarin employs constraint solving, which may be less efficient and does not always guarantee automatic termination, yet it eliminates spurious attacks and provides interactive proof guidance.
These fundamental differences render switching between tools nontrivial.

In this work, we present a sound translation from Tamarin to ProVerif that enables systematic comparative analysis of these tools.
Our translation provides a common ground for directly comparing their verification capabilities, identifying the core logic fragment they share, and understanding which formalism features have practical impact.
Our soundness guarantee ensures that any property verified in ProVerif within the faithful translation fragment also holds in the original Tamarin model.
We use a Tamarin model of Lowe's fix of the Andrew Secure RPC protocol as the case study to show the intended workflow: start with a previously published stateful Tamarin model, translate it automatically, and use ProVerif as a fast backend for non-injective and injective agreement.

Nevertheless, the fundamental differences between Tamarin and ProVerif preclude a complete translation. 
Tamarin's logic fragment is strictly more expressive than ProVerif's, necessitating careful handling of property specifications. 
Where feasible, we systematically rewrite Tamarin lemmas into logically equivalent ProVerif queries that preserve the intended semantics.
Certain Tamarin built-ins, specifically \texttt{multiset} and \texttt{bilinear-pairing}, admit no translation.
Diffie-Hellman support is partial, and XOR is translated only as a best-effort approximation outside the soundness theorem.

We evaluate our approach on 121 Tamarin models, addressing the following research questions:
\begin{enumerate}
	\item[RQ1] What common core logic fragment do Tamarin and ProVerif share for security properties?
	\item[RQ2] What novel encodings enable expressing Tamarin's features in ProVerif?
	\item[RQ3] To what extent do helping lemmas from Tamarin improve verification in ProVerif?
	\item[RQ4] What are the performance trade-offs between the tools on comparable models?
\end{enumerate}

\mypara{Contributions}
This work makes the following contributions:
\begin{itemize}
	\item A sound translation from Tamarin to ProVerif incorporating novel encodings for multiset rewrite semantics and simultaneous events
	\item A characterization of the common core logic fragment shared by both tools
	\item Systematic identification of semantic gaps and their practical implications
	\item A large-scale empirical study (121 models), including the Lowe case study, demonstrating broad agreement in verified results and practical performance gains
\end{itemize}

\mypara{Outline}
We discuss related work in \cref{sec:related-work} and provide background on Tamarin and ProVerif in \cref{sec:background}.
Our translation approach is presented in \cref{sec:translation}.
The soundness argument is developed in \cref{sec:soundness-completeness}.
Implementation details are provided in \cref{sec:implementation}.
We evaluate our approach in \cref{sec:evaluation} and conclude in \cref{sec:conclusion}.

\section{Related Work}
\label{sec:related-work}

We review prior work on formal verification tools, translations between their respective formalisms, relevant external feature extensions of ProVerif, and comparative analyses of verification tools.
\Cref{fig:related-work} gives an overview of the translations between these tools and situates our contribution.

\mypara{Protocol Verification Tools}
Over the years, several approaches have been proposed to formally model and verify security protocols (see \cite{barbosaSoKComputerAidedCryptography2021} for a survey).
Among the most widely used protocol verification tools are Tamarin~\cite{meier2013} and ProVerif~\cite{blanchet2001,blanchetSecurityProtocolVerifier2022}.
DeepSec~\cite{cheval2018} is used for bounded analysis of equivalence
properties such as privacy and anonymity.

DeepSec and ProVerif both employ a dialect of the applied\nobreakdash-$\pi$
calculus~\cite{proverif1stpi}, which extends Milner's $\pi$-calculus with the ability to
send free terms, making it well-suited for modeling the `perfect
cryptography' assumption of the Dolev-Yao model and concurrent communication
between parties.
In contrast, Tamarin's multiset-rewriting calculus is very similar to
term-rewriting, with Tamarin forcing a specific structure onto terms, describing
the overall system state as a multiset of `fact'-terms.
Term rewriting has also been used directly to analyze cryptographic
protocols in Maude-NPA~\cite{escobarMaudeNPACryptographicProtocol2009}.

\begin{figure}[t]
    \centering
    \begin{tikzpicture}[
        tool node/.style={
            draw, rounded corners, fill=white, font=\sffamily\bfseries\small, 
            inner sep=4pt, outer sep=2pt, align=center
        },
        arrow style/.style={
            -{Stealth[length=2.5mm, width=2.5mm]}, 
            very thick, %
            draw=black
        },
        citation label/.style={
            font=\sffamily\small, text=black
        },
        work label/.style={
            font=\sffamily\small, text=black
        }
    ]
        \node[tool node] (sapic) {SAPiC\textsuperscript{+}};
        \node[tool node, below=1.5cm of sapic] (deepsec) {DeepSec};
        \node[tool node, right=3cm of sapic] (tamarin) {Tamarin}; 
        \node[tool node, below=1.5cm of tamarin] (proverif) {ProVerif};

        \draw[arrow style] (sapic) -- node[citation label, fill=white, inner sep=1pt] {\cite{sapicplus}} (deepsec); 

        \draw[arrow style] (sapic) -- node[citation label, fill=white, inner sep=1pt] {\cite{sapicplus}} (tamarin); 

        \draw[arrow style] (sapic) -- node[citation label, fill=white, inner sep=1pt] {\cite{sapicplus}} (proverif); 

        \draw[arrow style] (tamarin) -- node[work label, fill=white, inner sep=1pt] {This work} (proverif); 

    \end{tikzpicture}
    \caption{Overview of related work and our contribution.}
    \label{fig:related-work}
    \Description{Diagram showing translation relationships between protocol verification tools. SAPiC+ translates to DeepSec, Tamarin, and ProVerif. This work adds a new arrow from Tamarin to ProVerif.}
\end{figure}
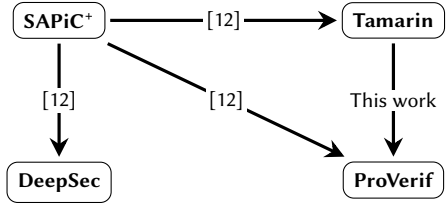

\begin{full}
\mypara{Translation from applied\nobreakdash-$\pi$ to Tamarin}
There are two prior works translating from applied\nobreakdash-$\pi$ to MSR---the opposite of the direction we are interested in.
\textcite{sapic} provide a translation from the \emph{Stateful Applied
Pi Calculus} (SAPiC), an extension to the applied\nobreakdash-$\pi$ calculus adding state access, to MSRs supported by Tamarin.
The authors suggest thinking of MSR as a low-level language that requires strong modeling
discipline and should ideally be produced by compiling from a more
high-level language.
Nevertheless, a substantial body of Tamarin models,
e.g.,
TLS~\cite{DBLP:conf/ccs/CremersHHSM17},
WPA2~\cite{DBLP:conf/uss/CremersKM20} and
Signal~\cite{cremersFormalAnalysisSessionHandling2023},
suggests that modelers are not at all
averse to using MSRs as an input language.
\textcite{sapicplus} extend this approach to a platform,
so that a single protocol specification (in a revised version of SAPiC, SAPiC+)
translates to
three different symbolic protocol verification tools: Tamarin,
ProVerif and DeepSec.
The differences between SAPiC+ and ProVerif, both descendants of the applied\nobreakdash-$\pi$ calculus,
mainly concern the
encoding of explicit state (in ProVerif) and term algebra with
destructor symbols (in SAPiC+).
Our translation of Tamarin's security properties to ProVerif queries builds on their work.
SAPiC+ allows embedding MSRs in a process,
but thus far only the translation to Tamarin supports this feature.
We thus integrate our MSR-to-ProVerif translation to add this missing piece
to the SAPiC+-to-ProVerif translation.

\mypara{Translation from MSR to applied\nobreakdash-$\pi$}
The only work translating from MSRs to applied\nobreakdash-$\pi$ was written well before Tamarin was developed.
\textcite{relatingmsrandpa} thus focus on a theoretical connection
with little regard to the size of the model or the efficiency of
verification.
They relate first-order multiset rewriting to an (unnamed) process algebra that
borrows from the $\pi$-calculus and the Calculus of Communicating Systems (CCS),
and they show
that
trace properties like secrecy and authentication
are preserved.
For the most basic interactions, i.e., network input and output, their
translation is similar to ours; however, it is limited
to MSRs where
\begin{enumerate*}
	\item roles can be clearly identified (i.e., no setup or out-of-band state synchronization);
	\item roles are linear (no conditionals, no state machines), and
	\item roles are finite (no replication).
\end{enumerate*}
Moreover, the intruder model is
considered to be part of the protocol description.
This is in contrast to Tamarin and ProVerif, where the intruder model is part of the semantics for optimization reasons.
In this early work, protocols were still analyzed by hand.
Altogether,
these limitations facilitate the translation and soundness proof to the point
where they share very little structure with our work.

In contrast to prior work, our approach soundly translates a large fragment of Tamarin's MSRs to ProVerif's applied\nobreakdash-$\pi$ calculus and scales to a large curated corpus of Tamarin examples.

\mypara{ProVerif Extensions}
ProVerif does not support the faithful representation of algebraic properties due to the lack of an associative function symbol.
This limits the modeling of certain cryptographic primitives, such as XOR operations and Diffie-Hellman exponentiation.
\citeauthor{proverif-xor}~\cite{proverif-xor, proverif-dh} describe a reduction step from protocols using XOR and Diffie-Hellman exponentiation to protocols without them, which can then be analyzed by ProVerif.
Since this reduction is a separate protocol-dependent step requiring user intervention, it is outside the scope of our automatic translation approach.
It is, however, orthogonal to our translation from Tamarin to ProVerif and can be applied to the translated ProVerif models as a post-processing step.

\mypara{Comparative Analysis of Verification Tools}
Although numerous protocol-verification tools have been developed, systematic comparative studies remain limited. 
Early pairwise comparisons examined specific tools like Scyther and ProVerif on selected protocols \cite{Dalal2010}, while performance-based evaluations compared execution time and memory consumption across multiple tools including Tamarin and ProVerif \cite{Lafourcade2015}. 
Recent comprehensive surveys have evaluated the latest versions of major tools \cite{Abed2024}, yet these studies often face challenges due to incompatible input languages and divergent modeling assumptions. 
While translation frameworks like \sapicp{}~\cite{sapicplus} now enable cross-tool verification, systematic model-level comparisons under semantically equivalent specifications remain scarce. 
Our translation-based approach addresses this gap by enabling direct, large-scale empirical comparison between Tamarin and ProVerif on semantically equivalent models.
\end{full}
\begin{conf}
\mypara{Translations and Extensions}
In the opposite direction, \sapic{} compiles a stateful applied\nobreakdash-$\pi$ calculus to Tamarin's MSRs, and \sapicp{} targets Tamarin, ProVerif, and DeepSec~\cite{sapic,sapicplus}; our property translation builds on \sapicp{}, while our integration supplies its previously missing ProVerif support for embedded MSRs.
While \sapicp{} treats MSR as a low-level target ideally produced by compilation from a higher-level language, the large curated corpus of hand-written Tamarin models we evaluate (\cref{sec:evaluation}) shows that modelers readily use MSRs directly as an input language, motivating a translation that takes them as its source.
Earlier work directly relating first-order MSR to a process algebra preserves trace properties but assumes identifiable, linear, finite roles, places the intruder in the protocol description, and analyzes protocols by hand~\cite{relatingmsrandpa}.
Our work instead provides an automated, sound translation of the supported Tamarin fragment that scales to a large curated corpus of Tamarin examples.
Existing protocol-dependent reductions eliminate XOR or Diffie-Hellman exponentiation before ProVerif analysis~\cite{proverif-xor,proverif-dh}.
They require user intervention and are outside our automatic translation, but can be applied as orthogonal post-processing.

\mypara{Comparative Analysis of Verification Tools}
Prior studies compare selected tools or protocols~\cite{Dalal2010,Lafourcade2015,Abed2024}, but incompatible languages and modeling assumptions hinder systematic model-level comparisons.
Our translation enables a direct, large-scale comparison of Tamarin and ProVerif on semantically equivalent models.
\end{conf}

\section{Background}
\label{sec:background}

We briefly recall the main concepts of Tamarin \cite{meier2013} and ProVerif \cite{blanchet2001}.
We present the two formalisms separately and emphasize the translation links between MSRs and processes, Tamarin action facts and ProVerif events, as well as lemmas and queries.

\subsection{Tamarin}

\mypara{Terms}
In Tamarin, messages are represented by abstract terms composed of function symbols that represent cryptographic primitives.
The terms are drawn from an \emph{order-sorted term algebra}, where each term is assigned a sort.
The sorts form a partially ordered set with a top sort, \msgsort, and two incomparable subsorts: \freshsort{} and \pubsort.

The sort \freshsort{} (fresh names) models random values such as keys and nonces.
The sort \pubsort{} (public names) models known constants, such as the identities of protocol parties.
For each sort $s$, there is a countably infinite set of variables $\Vars_s$.
The union of all such sets is the set of variables $\Vars$.
Keys and nonces are constants from a countably infinite set of
names $\Names$, divided into public (attacker) names $\PN$
and secret (protocol) names $\FN$.
Given a set of variables $\Vars$, the set of terms $\Terms(\Sign,\PN,\FN,\Vars)$, short $\Terms$,
is thus constructed over
$\Names$ and $\Vars$ and applications of function
symbols in the signature $\Sign$ on terms. Let $f \in \Sign^n$ denote a function
symbol with arity $n$, alternatively written as $f/n$.
When we apply a
\emph{substitution} $\sigma$, i.e., a well-sorted partial function from $\Vars$ to $\Terms$,
to a term,
we substitute each variable $v$ in the domain of $\sigma$ by the term $\sigma(v)$.
Let $\vars(t)$ denote the set of variables occurring in a term $t$ and $\vars^+(t)$ denote the set of variables occurring in $t$ if $t$ is not a variable, i.e., $\vars^+(t) = \emptyset$ if $t \in \Vars$ and $\vars^+(t) = \vars(t)$ otherwise.

To equip terms with a meaning, Tamarin allows the user to define
a set
of equations $\ET \subseteq \Terms \times \Terms$, for instance,
$\sdec(\senc(y, x), x) = y$ to model encryption.
Tamarin then considers terms modulo the equivalence relation $=_\ET$ defined as the smallest equivalence relation that
\begin{enumerate*}[label=(\arabic*)]
  \item contains $\ET$,
  \item is closed under application of function symbols, and
  \item is closed under substitution of variables by terms.
\end{enumerate*}
For instance, $ f(\sdec(\senc(g(m),k),k)) =_\ET f(g(m))$.

\mypara{Protocol Model}
The protocol's behavior is modeled
as a system of labeled multiset rewrite rules (MSRs)
of the form
$\ru = l \msrewrite{a} r$,
where
$\prems(\ru)=l$ are the premises of the rule,
$\concls(\ru)=r$ the conclusions
and
$\acts(\ru)=a$ the actions that label the rule.
The facts in $a$ are called \emph{action facts} and annotate the trace.
These rules rewrite a multiset of
ground facts from the set $\GFacts$, where a fact
$\fact{F}(t_1, \ldots, t_k)$ of arity $k$ is ground if all $k$ terms $t_1, \ldots, t_k$ are ground and terms are ground if they contain no
variables.
There are predefined fact symbols for special purposes.
The fact that a name $n\in\FN$ is fresh (i.e., was not used before)
is marked with the unary fact symbol $\Fr$.
Terms being received or sent are represented by $\In$ and $\Out$, respectively.
\begin{example}\label{ex:msr}
  \begin{equation*}
    \mset{ \pfact{Ltk}(A, k), \In(x) } \msrewrite{\fact{Enc}(x,k)} \mset{ \Out(\senc(x, k)) }
  \end{equation*}
  This rule models that agent $A$
  receives a message from the network, encrypts it with its long-term key $k$, and transmits the ciphertext.
\end{example}

An exclamation mark in front of a fact symbol indicates that it is
\emph{persistent} and can be consumed arbitrarily often. For example,
freshness $\Fr$ is a linear fact, whereas $\pfact{Ltk}$ is a persistent
fact.
In \cref{ex:msr}, $\pfact{Ltk}(A,k)$ is therefore a persistent fact recording that $k$ is $A$'s long-term key, while the action fact $\fact{Enc}(x,k)$ labels the encryption step.

\mypara{Traces}
For a fact multiset $M$, let $\lfacts(M)$ and $\pfacts(M)$ denote its linear and persistent facts, respectively, and let $\mathit{set}(M)$ discard multiplicities.
We mark multiset inclusion, difference, and union with a $\#$ superscript: $\msubset$, $\msetminus$, and $\mcup$.
From here on, $\rules$ denotes a finite set of Tamarin MSRs, $\ru=l\msrewrite{a}r$ a single rule, $\GFacts$ the set of ground facts, $\GFacts^\#$ the multisets of ground facts, and $\ginsts_\ET(\rules)$ the set of ground instances of $\rules$ modulo an equational theory $\ET$.

We define a labeled transition relation $\rightarrow_{\ginsts_{\ET}(\rules)} \subseteq \GFacts^{\#} \times \GFacts^{\#} \times \GFacts^{\#}$
as a rewriting step between multisets of ground facts, with $S\in\GFacts^\#$ denoting the current state.
Rewriting is possible if all facts in $l$ are in $S$
and replaces all linear facts in $l$ with all facts in $r$.
Formally:
\begin{equation*}
  \infer
  {S \xrightarrow{a}_{\ginsts_\ET(\rules)} (S \msetminus \lfacts(l)) \mcup r}
  {l\msrewrite{a}r \in \ginsts_\ET(\rules) & \lfacts(l) \msubseteq S & \pfacts(l) \subseteq \mathit{set}(S)}
\end{equation*}
\begin{figure}[t]
  \begin{align}
    \mset{\Out(x)}                    &\withmsr{\hphantom{\K(x)}} \mset{!\K(x)}                    \tag{\textsc{MDOut}}   \\
    \mset{\Fr(x:\mathit{fresh})}      &\withmsr{\hphantom{\K(x)}} \mset{!\K(x:\mathit{fresh})}     \tag{\textsc{MDFresh}} \\
    \msempty{}                        &\withmsr{\hphantom{\K(x)}} \mset{!\K(x:\mathit{pub})}       \tag{\textsc{MDPub}} \\
    \mset{!\K(x_1), \ldots, !\K(x_k)} &\withmsr{\hphantom{\K(x)}} \mset{!\K(f(x_1, \ldots, x_k))}  \tag{\textsc{MDApp}} \\[-2pt]
    \mset{!\K(x)}                     &\withmsr{\K(x)} \mset{\In(x)}                    \tag{\textsc{MDIn}}
  \end{align}
  \caption{The set of message deduction rules \textsc{MD}.}
  \label{fig:message-deduction-rules}
  \Description{Five multiset rewrite rules defining Dolev-Yao attacker capabilities: MDOut captures network outputs, MDFresh and MDPub add fresh and public names to attacker knowledge, MDApp applies function symbols, and MDIn injects known terms into the network.}
\end{figure}

Tamarin combines user-defined protocol rules with the built-in message-deduction rules $\MD$ shown in~\cref{fig:message-deduction-rules}.

These rules represent a standard Dolev-Yao attacker who obtains knowledge ($!\K$) by eavesdropping on the network ($\Out$), creating fresh names ($\Fr$), or using public values.
The knowledge can be used to construct new terms by applying function symbols.
Network communication is mediated by $\Out$ and $\In$: \textsc{MDOut} turns a protocol output $\Out(x)$ into attacker knowledge $!\K(x)$, and \textsc{MDIn} lets the attacker provide any known term to a protocol rule by producing $\In(x)$. The action fact $\K(x)$ on \textsc{MDIn} records this use of attacker knowledge in the trace.
Executions and traces are accordingly defined over $\ginsts_\ET(\rules \cup \MD \cup \set{\Fresh})$, where the built-in rule $\Fresh\colon \msempty \withmsr{} \mset{\Fr(x:\mathit{fresh})}$ creates fresh names; to ease notation, we keep writing $\rules$ in the subscripts.

\begin{definition}[Executions and Traces]
  \label{def:executions}
  \label{def:traces}
  An execution of a set of rules $\rules$
  is a sequence of transitions from multisets of ground facts
  labeled by action-fact multisets starting with the empty multiset
  \begin{equation*}
    \msempty \withmsr{a_1}_{\ginsts_\ET(\rules)} S_1  \cdots \withmsr{a_n}_{\ginsts_\ET(\rules)} S_n\,,
  \end{equation*}
  such that no fresh name is chosen twice, i.e.,
  for the special fact symbol $\Fr$, we have
  $S_{i+1} \msetminus S_i = \mset{\Fr(n)}$
  and
  $S_{j+1} \msetminus S_j = \mset{\Fr(n)}$
  implies $i=j$.
  The set of executions is denoted by $\execmsr(\rules)$.
  The set of traces $\tracesmsr(\rules)$ is the projection on the action labels
  in $\execmsr(\rules)$ including empty action multisets.
  The set of weak traces $\wtracesmsr(\rules)$ is also the projection on the action labels
  in $\execmsr(\rules)$ but excluding empty action multisets.
\end{definition}

\begin{definition}[Reachable state]
  A state $\mss$ is \emph{reachable} if it can be reached from the initial state $\mss[0]$ by a chain of $\withmsr{a}_{\ginsts_\ET(\rules)}$ transitions:
  \begin{align*}
    \mss[0] \withmsr{a_1}_{\ginsts_\ET(\rules)} \ldots \withmsr{a_n}_{\ginsts_\ET(\rules)} \mss[n] = \mss\,.
  \end{align*}
\end{definition}

\begin{figure*}[t]
  \setlength{\grammarindent}{3em}
  \setlength{\grammarparsep}{4pt plus 1pt minus 1pt}
  \begin{subfigure}[b]{0.33\linewidth}
    \centering
    \begin{minipage}[t]{.5\linewidth}
      \begin{grammar}
         <$F$> ::= $I$ [ $\Leftrightarrow$ $I$ ]

         <$I$> ::= $D$ [ $\Rightarrow$ $I$ ]

         <$D$> ::= $C$ $(\lor \; C)^*$

         <$C$> ::= $N$ $(\land \; N)^*$

         <$N$> ::= $[\neg]$ $A$

         <$A$> ::= $(F)$
         \alt $[\exists \mid \forall]$ $l^+.$ $F$
         \alt $\fact{F}(t_1, \ldots, t_n)@i$
         \alt $t = t'$
         \alt $i \dotless i'$
         \alt $i \doteq i'$
        \alt $\bot$

         <$l$> ::= $i$ \alt $w$
      \end{grammar}
    \end{minipage}
    \caption{Tamarin's trace formula grammar.}
    \label{fig:tamarin-trace-grammar}
  \end{subfigure}%
  \hfill
  \begin{subfigure}[b]{0.33\textwidth}
    \centering
    \begin{minipage}[t]{.75\linewidth}
      \begin{grammar}
        <$q$> ::= $F \land \ldots \land F$
        \alt      $F \land \ldots \land F \implies H$

        <$H$> ::= $F$
        \alt      $H \land H$
        \alt      $H \lor H$
        \alt      $F \implies H$
        \alt      $C$
        \alt      $\bot$

        <$F$> ::= $\fact{F}(t_1, \ldots, t_n)@i$
        \alt      $\Patt{t}@i$

        <$C$> ::= $t \circ_e t'$  with $\circ_e \in \{ =,\neq \}$
                      \alt $i \circ_t i'$ with $\circ_t \in \{ \dotleq , \doteq , \dotless , \not \doteq\}$
      \end{grammar}
    \end{minipage}
    \caption{ProVerif's trace formula grammar.}
    \label{fig:proverif-trace-grammar}
  \end{subfigure}%
  \hfill
  \begin{subfigure}[b]{0.33\textwidth}
    \centering
    \begin{minipage}[t]{\linewidth}
      \begin{grammar}
        <$P$, $Q$> ::=\\
        $\Pnil$
        \alt $\Pout(c, t); \; P$
        \alt $\Pin (c, x); \; P$
        \alt $P \mid Q$
        \alt $!P$
        \alt $\Pnew \; n; \; P$
        \alt $\Plet \; x = t \; \Pin \; P \; \Pelse \; Q$
        \alt $\Pif \; E \; \Pthen \; P \; \Pelse \; Q$
        \alt $\event{f_e(t_1, \ldots, t_n)}; \; P$
        \alt $\Pinsert\ \mathit{tbl}(v_1, \ldots, v_n); \; P$
        \alt $\Pget\ \mathit{tbl}(p_1, \ldots, p_n)\ \Pin\ P\ \Pelse\ Q$
      \end{grammar}
    \end{minipage}
    \caption{ProVerif's process calculus grammar.}
    \label{fig:pc}
  \end{subfigure}
  \caption{Grammars for Tamarin and ProVerif. Tamarin's full trace formula logic
is shown to illustrate its expressiveness. ProVerif's grammar shows the SA-QBF
  fragment used for correspondence and reachability queries, omitting phases
  and native injective correspondences. Injective agreement is instead handled
  by the two-query SA-QBF encoding of \cref{sec:lemma-translation}.}
  \label{fig:trace-property-grammar}
  \Description{Three side-by-side grammar specifications: (a) Tamarin's trace formula grammar with nested quantifiers and temporal operators, (b) ProVerif's SA-QBF fragment with single alternation structure, and (c) ProVerif's process calculus including parallel composition, replication, and table operations.}
\end{figure*}

\mypara{Trace Properties}
Properties are expressed as \emph{lemmas} in a two-sorted first-order trace logic with quantification over messages and timepoints.
The message sort \msgsort{} ranges over terms, whereas the timepoint sort \tempsort{} ranges over timepoints.
Trace atoms have the form $\fact{F}(\vec{t})@i$, referring to Tamarin action facts at trace timepoints.
Tamarin supports universal (\emph{all-trace}) properties, which must hold on every trace, and existential (\emph{exists-trace}) properties, which require only one witness trace.
The Lowe case study in \cref{sec:translation} uses all-trace secrecy and agreement properties.
\begin{definition}[Tamarin's two-sorted first-order logic]
  \label{def:tam-logic}
  Let $i, i' \in \Vars_{temp}$ (timepoint variables), $w \in \Vars_{msg}$ (term variables), and $f \in \Sign^n$.
  Furthermore, $t_1, \ldots, t_n \in \Terms$.
  Tamarin's trace formula logic is defined by the grammar in \cref{fig:tamarin-trace-grammar}.
\end{definition}

To determine if a trace formula $\varphi$ holds,
Tamarin converts it into an equivalent \textit{guarded trace property}.
First, $\varphi$ is transformed into \textit{negation normal form} (NNF),
where negations apply only to trace atoms, and the remaining operations are $\land$, $\lor$, $\exists$, or $\forall$.
A formula in NNF is a \emph{guarded trace formula} if it has the form $\exists \vec{x}. \; (\fact{F}@i) \land \psi$ or $\forall \vec{x}. \; \neg (\fact{F}@i) \lor \psi$,
where $\psi$ is itself a guarded trace formula
and the trace atom $\fact{F}@i$, called the guard,
ensures that all quantified variables are instantiated with terms or indices that occur in the trace,
$\vec{x} \subseteq \vars(\fact{F}@i) \cap (\Vars_{msg} \cup \Vars_{temp})$.
A guarded trace formula that contains only terms in $\Vars \cup \PN$
and where all variables are quantified is called
a \emph{guarded trace property}.

\subsection{ProVerif}

\mypara{Terms}
In ProVerif, cryptographic primitives are specified as constructor symbols,
which correspond directly to Tamarin's function symbols.
ProVerif supports \emph{destructors},
which can also be used to define cryptographic operators and the relationships between them,
much like, and as an alternative to, equations in Tamarin.
We define the set of terms $\Terms$ to contain all constructor terms,
i.e., all terms without destructor symbols.
This definition aligns with the previous definition of $\Terms$ for Tamarin.
Like in Tamarin, equational theories can be defined as a
finite set of equations $\ET \subseteq \Terms \times \Terms$.
By contrast,
a destructor $\mathit{g}$ is a partial function $\Terms^n \rightarrow \Terms$ (for some arity $n \in \mathbb{N}$) that a process can apply.
The semantics are provided by an ordered list of rewrite rules $\ddef(\mathit{g})$, each of the form $\mathit{g}(t_1, \ldots , t_n) \rightarrow t$.
When the process encounters an instance of the term $\mathit{g}(t_1, \ldots, t_n)$,
it tries to rewrite it according to one of the rules in $\ddef(\mathit{g})$.
If it succeeds, it returns the term $t \in \Terms$.
If none of the rules are applicable, the destructor application fails and the process blocks.
Destructors can also be defined as private, in which case they can occur in the process but are not available to the attacker.

\begin{full}
\begin{example}
  Let the constructor $\senc/2$ model symmetric encryption.
  Decryption can be described
  with an additional constructor $\sdec/2$ and the equation
  $\sdec(\senc(m,k),k)=_\ET m$, as in Tamarin.
  Alternatively, decryption can be described with a
  \emph{destructor} $\sdec/2$ defined through the rewrite rule
  $\sdec(\senc(m,k),k) \to m$.
  If we supply an invalid key to the destructor, it will fail and block the execution,
  whereas $\sdec(\senc(m,k),k')$ with $k\neq k'$ is a valid term, albeit irreducible.
  Hence, the destructor can model an encryption scheme that detects valid ciphertexts.
\end{example}
\end{full}
\begin{conf}
\begin{example}
  Symmetric decryption can be modeled either with a constructor $\sdec/2$ and the equation $\sdec(\senc(m,k),k)=_\ET m$, as in Tamarin, or with a \emph{destructor} $\sdec/2$ given by the rewrite rule $\sdec(\senc(m,k),k) \to m$.
  The destructor blocks on an invalid key, modeling an encryption scheme that detects valid ciphertexts.
\end{example}
\end{conf}

\mypara{Protocol Model}
Both the implementation of ProVerif and its underlying calculus have gone through several iterations.
For the purpose of this work, it suffices to consider the core language as described in~\cite{BlanchetFnTPS16}.
Its syntax is shown in~\cref{fig:pc}.
For readability, we omit types in the presentation. The implementation
uses the default type $\mathtt{bitstring}$ throughout.

The empty process $\Pnil$ terminates a process definition.
In the grammar, a semicolon separates a construct from its continuation; the calculus has no general sequential composition.
The process $\Pout(c, t); \; P$ outputs $t$ on channel $c$ and continues as $P$, while $\Pin(c, x); \; P$ receives a message from $c$, binds it to $x$, and continues as $P$.
The process $P \mid Q$ runs $P$ and $Q$ in parallel, and the \textit{replication} $!P$ runs arbitrarily many copies of $P$.
The \textit{restriction} $\Pnew \; n; \; P$ creates a fresh name $n$ scoped to $P$ and then executes $P$.
The construct $\Plet \; x = t \; \Pin \; P \; \Pelse \; Q$ evaluates $t$: on success, it binds the result to $x$ and executes $P$; otherwise, it executes $Q$.
Evaluation can fail because of an incorrect destructor application.
The construct $\Pif \; E \; \Pthen \; P \; \Pelse \; Q$ similarly executes $P$ when $E$ evaluates to $\mathtt{true}$ and $Q$ otherwise.
In both constructs, an omitted $\Pelse$ branch means $Q=\Pnil$, so failed evaluation terminates execution.
The $\eventC$ construct enables the specification of trace properties.
\textit{Events} are written as $\fact{F}(t_1, \ldots, t_n)$ where $\fact{F} \in \FSign^n$ is an event symbol and $t_1, \ldots, t_n$ are terms.
They do not influence the execution of a process.
Events record that a certain point in the protocol execution was reached
with certain values, specified as the event's arguments.

Tables provide a mechanism for persistent storage in ProVerif models.
Once a value has been added to a table, it cannot be deleted or replaced.
Tables are declared with a number of columns $n$;
processes can append a row of $n$ values (with $\Pinsert$) or
retrieve an arbitrary row that matches a pattern (with $\Pget$). If there is a match, $P$ is executed, otherwise the (optional) process $Q$ is executed.
Note that the declared tables are not accessible by the attacker.

\mypara{Traces}
We define a labeled transition relation $\rightarrow_{pv} \subseteq \pvconf \times \mathcal{L} \times \pvconf$,
where $\pvconf$ denotes a process configuration and $l \in \mathcal{L}$ a label,
as a reduction relation between process configurations.
A process configuration $\pvconf = \pvconfquad$ is defined in a way similar to~\cite[Appendix F.1, Fig. 8]{sapicplus-full}, where
\begin{itemize}
  \item \E contains the free names of \Pro and names created during execution by the protocol or attacker;
  \item \Pro is a multiset of processes representing the current state of the process;
  \item \Tbl is the set of table entries $T(M_1, \ldots, M_n)$ inserted during the process execution up to this point; and
  \item \Att is a set representing the attacker knowledge.
\end{itemize}
The set \E is partitioned into the disjoint sets $\mathcal{N}_{pub}$ and $\mathcal{N}_{priv}$,
for public (attacker) names and private (protocol) names, respectively.
We transition between process configurations by applying reduction rules from~\cite[Figure~8]{sapicplus-full}.
The label $l$ is either empty ($\epsilon$), an event, or a $\msgC$ label.
Executions of $\Pout$ and $\Pnew$ add terms to the attacker knowledge.
Notably, the set \Att can be saturated by applications of the \textsc{App} rule.
This rule corresponds to the attacker obtaining terms by applying constructors and destructors to terms from its knowledge set.

\begin{definition}[Executions and Traces]
  \label{def:pv-execs-trcs}
  A ProVerif execution is a sequence of reduction rule applications
  starting with the initial process configuration $\pvconf[0] = (\E[0], \{P\}, \emptyset, \Att[0])$
  \begin{equation*}
    \pvconf[0] \withpv{l_1} \pvconf[1] \cdots \withpv{l_n} \pvconf[n].
  \end{equation*}
  The transitions in the execution are annotated with the labels of the applied rules.
  The set of executions is denoted by $\execpv(P)$.
  The set of traces $\tracespv(P)$ is the projection on the labels
  in $\execpv(P)$ including empty labels ($\epsilon$).
  The set of weak traces $\wtracespv(P)$ is also the projection on the labels
  in $\execpv(P)$ but excluding the empty labels $\epsilon$.

  In contrast to Tamarin, which annotates transitions by multisets of action facts, ProVerif uses atomic labels such as events and message labels.
  To ease the comparison of both tools, we implicitly lift ProVerif's atomic labels to singleton multisets of labels.
\end{definition}

\begin{definition}[Reachable configuration]
  A ProVerif configuration $\pvconf$ is \emph{reachable} if there is a chain of $\withpv{l}$ transitions starting at the initial configuration $\pvconf[0]$ such that
  \begin{align*}
    \pvconf[0] \withpv{l_1} \ldots \withpv{l_n} \pvconf[n] = \pvconf\,.
  \end{align*}
\end{definition}

\mypara{Trace Properties}
In ProVerif, trace properties are categorized into \textit{reachability} and \textit{correspondence} properties, both of which start with the keyword $\query$ (see~\cref{fig:proverif-trace-grammar}).
A query contains events emitted by the process or attacker predicates $\Patt{t}@i$, stating that the attacker can derive term $t$ at timepoint $i$.
Reachability properties are of the form
\begingroup
\addtolength{\belowdisplayskip}{2pt}
\begin{equation*}
  \query \; \fact{F}_1 \; \&\& \; \ldots \; \&\& \; \fact{F}_n,
\end{equation*}
\endgroup
where $\fact{F}_i$ are events.
They ensure the existence of at least one trace in which all the events $\fact{F}_1, \ldots, \fact{F}_n$ occur and are hence comparable to Tamarin's exists-trace lemmas.
Correspondence properties are of the form $\query \; \fact{F}_1 \; \&\& \; \ldots \; \&\& \; \fact{F}_n \longrightarrow H$,
where $\fact{F}_1, \ldots, \fact{F}_n$ are events and the \textit{conclusion} $H$ is a quantifier-free trace formula over trace atoms.
They are universally quantified over all traces and are comparable to Tamarin's all-trace lemmas.
Thus, a correspondence query reads: whenever the premise events $\fact{F}_1,\ldots,\fact{F}_n$ occur together, the conclusion $H$ must also be justified in the same trace.
The disjunctive normal form of $H$ must not contain negated events.

Internally, ProVerif translates reachability queries into correspondence queries with the conclusion $\bot$, so all queries can be treated as implications.
In contrast to Tamarin, ProVerif only allows a single quantifier alternation.
The variables occurring in the premise are universally quantified,
while the variables occurring in the conclusion but not in the premise are existentially quantified.
This convention is the main syntactic difference the lemma translation must manage: Tamarin exposes quantifiers explicitly, while ProVerif infers them from where variables occur in the query.

\begin{definition}[SA-QBF]
  \label{def:sa-qbf}
  \Cref{fig:proverif-trace-grammar} presents the syntax
  of ProVerif's logic fragment we target, which we call
  \emph{SA-QBF} (single-alternation quantified Boolean formulas).
  This fragment captures ProVerif's correspondence and reachability queries,
  excluding phases and native injective correspondences.
  The translation of injective agreement remains in this fragment by decomposing
  the source property into correspondence and uniqueness queries
  (\cref{sec:lemma-translation}).
  When a timepoint inequality between $i$ and $j$ occurs in the conclusion of a formula, where $i$ and $j$ are the time variables associated with facts $\fact{F}$ and $\fact{G}$, respectively, the following two conditions must hold:
  \begin{enumerate*}[label=(\arabic*)]
    \item $\fact{F}@i$ occurs in the premise or $\fact{F}$ is an event;
    \item $\fact{G}@j$ occurs in the conclusion or $\fact{G}$ is an event.
  \end{enumerate*}
\end{definition}

\begin{remark}
  \label{rem:sa-qbf-restrictions}
  ProVerif queries must conform to SA-QBF as presented in \cref{fig:proverif-trace-grammar}.
  In particular:
  \begin{enumerate*}[label=(\arabic*)]
    \item reachability queries do not support disjunctions,
    \item correspondence queries do not support conjunctions at the top level, and
    \item negation is not supported in queries.
  \end{enumerate*}
  These restrictions necessitate query rewriting and splitting for certain
  Tamarin lemmas (see \cref{sec:lemma-translation}).
  ProVerif also supports phases and native injective correspondences beyond
  SA-QBF; neither is needed for the exact encodings used here.
\end{remark}

\section{Translation}
\label{sec:translation}

Our translation serves two purposes:
\begin{enumerate*}[label=(\arabic*)]
  \item enabling practical use of ProVerif on Tamarin models, and
	\item providing a rigorous foundation for comparative analysis.
\end{enumerate*}
The translation challenges we address reveal fundamental differences between the formalisms.
We introduce encoding techniques that bridge these gaps where possible and precisely characterize the remaining limitations.

The translation is composed of three main parts:
  process generation,
  property translation, and
  bridging semantic gaps between the two tools.
\Cref{tab:syntax-map} summarizes the main syntactic correspondences. The remainder of this section expands these correspondences and explains the cases where they require additional encoding.

\subsection{Rule Translation}
\label{sec:rule-translation}

Given a set of MSRs $\rules = \{\ru_1,\ldots,\ru_n\}$, we translate each rule to a subprocess in ProVerif.
These subprocesses run in parallel and are replicated to allow multiple applications.
The translated process of a rule $\ru$ is denoted by $\tran{\ru}$.
At a high level, premises of a Tamarin rule become ProVerif inputs or table lookups, action facts become ProVerif events, and conclusions become outputs or table insertions.

\begin{table}[t]
  \caption{Correspondence between Tamarin and ProVerif constructs. The channel $c$ is a global public channel in ProVerif.}
  \label{tab:syntax-map}
  \small
  \renewcommand{\arraystretch}{1.1}
  \begin{tabularx}{\linewidth}{>{\raggedright\arraybackslash}X>{\raggedright\arraybackslash}X}
    \toprule
    \textbf{Tamarin source construct} & \textbf{ProVerif target construct} \\
    \midrule
    Storage fact $\fact{F}(\vec{t})$ & Table entry $\fact{F}_{tbl}(\vec{t})$ \\
    Premise fact $\fact{F}(\vec{t})$ & $\Pget\ \fact{F}_{tbl}(\vec{t})$ \\
    Conclusion fact $\fact{F}(\vec{t})$ & $\Pinsert\ \fact{F}_{tbl}(\vec{t})$ \\
    $\In(t)$ & $\Pin(c,t)$ \\
    $\Out(t)$ & $\Pout(c,t)$ \\
    $\Fr(n)$ & $\Pnew\ n$ \\
    Action fact $\fact{F}(\vec{t})$ & Event $\event{\fact{F}(\vec{t})}$ \\
    All-trace lemma & Correspondence query \\
    Exists-trace lemma & Reachability query \\
    \bottomrule
  \end{tabularx}
\end{table}

We then translate $\rules$ to the process
\begin{equation*}
	\tran{\rules} \defas{} !\mkern2mu\tran{\ru_1} {}\mid{} \cdots {}\mid{} !\mkern2mu\tran{\ru_n}\,.
\end{equation*}
The replication of each subprocess is an over-approximation, but we do not lose precision, since ProVerif's Horn clause abstraction acts as an implicit top-level replication, which can be pushed down the parallel composition.

Each subprocess emulates Tamarin MSR semantics using ProVerif's table construct.
For each fact symbol $\factsymbol{F}$ with arity $n$ in Tamarin, we create a corresponding ProVerif table $\factsymbol{F}_{tbl}$ with the same arity (the multiset encoding of \cref{sec:bridging-gaps} adds one further identifier column).
Modifications of the protocol state in Tamarin through facts are then emulated by inserting and looking up entries in these tables.
Since tables in ProVerif are monotonic, this is an over-approximation.
We discuss in \cref{sec:bridging-gaps} how we overcome this limitation.

The special fact symbols $\In$, $\Out$ and $\Fr$, whose semantics are defined through the built-in message deduction rules \textsc{MD} in \cref{fig:message-deduction-rules}, are treated differently.
These facts have direct counterparts in ProVerif: $\Pin$, $\Pout$ and $\Pnew$.

Tamarin supports pattern matching on facts in a rule's premise
\begin{equation*}
  \msrulearrowempty{\factsymbol{F}(h(x))}{}{\Out(x)}\,.
\end{equation*}
In contrast, ProVerif permits only tuple patterns and equality tests in inputs ($\Pin$) and table lookups ($\Pget$), but no matching under other constructors such as $h(x)$.
To emulate this behavior, we introduce private auxiliary destructors that extract variables from patterns.
These destructors are translation helpers used only inside the generated process; they are private, not attacker operations, and not additional cryptographic capabilities.
If exposed to the attacker, they would otherwise allow them, for example, to look inside encrypted messages.

We create such destructors for each variable $v$ and top-level function application $t$ occurring in a pattern in the premise of a rule.
Let $g_{v, t}$ denote the destructor extracting $v$ from $t$ with the rewriting rule $g_{v,t}(t) \rightarrow v$.
For the above example, we introduce the destructor $g_{x,h(x)}$ and obtain
\begin{align*}
 \Pget\ &\fact{F}_{tbl}(\aux_{h(x)})\ \Pin \\
       &\Plet\ x = g_{x,h(x)}(\aux_{h(x)})\ \Pin \\
       &\Pout(c, x)
\end{align*}
for an auxiliary variable $\aux_{h(x)}$ not occurring elsewhere.

To formally define the translation, we introduce a sequencing
operator that applies a process expression to each element of an ordered multiset, using a fixed lexicographical order for names, terms, and facts.

\begin{definition}[Sequencing Operator]
	Let $X$ be an ordered multiset, $e \colon X \to \mathit{Process}$ a process expression, and $P$ a follow-up process.
	The sequencing operator $\seq$ is defined as
	\begin{equation*}
		\smashseq{x \melem X}\; e(x); P \defas
		\begin{cases}
			P                                                         & \text{if $X = \emptyset$} \\
			e(y); \smashseq{x \melem X\msetminus \mset{y}}\; e(x); P & \text{if $y = \min X$}
		\end{cases}
	\end{equation*}
\end{definition}
For example, for the names $a,b,c\in\Names$, we have
\begin{equation*}
  \smashseq{x \melem \{a,b,c\}}\; \Pnew~x;\Pnil = \Pnew~a;\Pnew~b; \Pnew~c; \Pnil
\end{equation*}

Next, let $\aux(t)$ be $t$ if $t \in \Vars$ and a fresh auxiliary variable $\aux_{t}$ otherwise.
We can define the translation of \Storage facts (all facts except $\In$, $\Out$, and $\Fr$) and input facts that may contain patterns.

\begin{definition}[Translation of \Storage facts]
	\begin{align*}
		\tran{ \factsymbol{F}(t_1,\ldots,t_n)}; P := \Pget{}\ &\fact{F}_{tbl}(\aux(t_1),\ldots,\aux(t_n))~\Pin \\
		\smashseq{1 \leq i \leq n,\, v \in \vars^+(t_i)}\; &\Plet~v = g_{v,t_i}(\aux_{t_i})\ \Pin \\
		&P
	\end{align*}
\end{definition}

For input and output facts, let $c$ be the single public channel used for all communication to emulate Tamarin's global network model.
\begin{definition}[Translation of input facts]
	\begin{align*}
		\tran{\In(t)}; P {}:={} &\Pin~(c, \aux(t)); \\
		\smashseq{v \in \vars^+(t)}\; &\Plet~v = g_{v,t}(\aux(t))\ \Pin \\
		&P
	\end{align*}
\end{definition}

Observe that the $\Pget$ and $\Pin$ processes have no $\Pelse$ branches. This emulates the behavior of a rule being applicable only if all facts of its premise satisfy the patterns and are present in the current protocol state.
ProVerif denotes variables that are matched instead of bound by prepending the variable name with `='.
This becomes relevant when a variable occurs both in a pattern and directly in the fact.
In this case, we prefix it with `=' in the $\Pget$ or $\Pin$ to indicate that it is matched, as shown in \cref{fig:translation-ex}.

The two remaining parts of the translation are action facts and the handling of public variables.
Action facts are directly translated to ProVerif events.
Tamarin allows public variables (variables of sort $\pubsort$) in action facts and $\Out$ facts that are not bound in the premise of a rule, which ProVerif does not support.
We define a generic function $\pubv$ that returns the set of public variables in a term or fact as a term list.
We prepend $\Pin(c,\pubv(\cdot))$ to action facts and $\Out$ facts such that the ProVerif attacker can instantiate the public variables.
In Tamarin, such variables range over public names, which are always deducible; the added inputs thus make every Tamarin instantiation of these variables available in the translated process.
Further attacker-chosen instantiations only enlarge the trace set, consistent with the translation's over-approximating design.
Each input is placed immediately before the event, insertion, or output whose public variables it instantiates.

\begin{definition}[Rule Translation]
	\label{def:tran-rule}
	A rule $\ru = l \msrewrite{a} r$
	is translated to the following process.
	\begin{align*}
	    \tran{\ru} \coloneq{} & \smashseq{\fact{F} \melem l|_{\Storage, \In}}\; \tran{\fact{F}}; \\
	                           & \smashseq{\Fr(n) \melem l}\; \Pnew~n; \\
	                           & \smashseq{e \melem a}\; (\Pin(c,\pubv(e)); \Pevent~e); \\
	                           & \smashseq{\fact{F} \melem r|_{\Storage}}\; (\Pin(c,\pubv(\fact{F})); \Pinsert~\fact{F}); \\
	                           & \smashseq{\Out(t) \melem r}\; (\Pin(c,\pubv(t)); \Pout(c,t)); \\
	                           & \Pnil
	\end{align*}
\end{definition}

The translated process first checks for the presence of all facts in the premise of the rule.
Then, it draws all fresh names specified by $\Fr$ facts.
Next, it translates each action fact by first inputting the public variables and then raising the corresponding ProVerif event.
Finally, it applies all conclusions by inserting the corresponding table entries and performing all outputs, again inputting public variables first.
This order ensures that ProVerif emits events only after the premises are met and never produces an output or table insertion before the corresponding action events.

\newsavebox{\tamarinbox}
\newsavebox{\proverifbox}

\begin{lrbox}{\tamarinbox}
  \begin{lstlisting}[language=spthy,basicstyle=\small\ttfamily,linewidth=\linewidth]
rule A_2:
  [ !Key($A,$B,~sk),
    StateASend($A,$B,~sk,~na),
    In(senc(<'2',~na,kabp,$B>,~sk)) ]
--[ Secret($A,$B,kabp),
    Commit_A($A,$B,<'A','B',~sk,kabp>) ]->
  [ Out(senc(<'3',~na>,kabp)),
    StateAReceive($A,$B,~sk,~na,kabp) ]
  \end{lstlisting}
\end{lrbox}

\begin{lrbox}{\proverifbox}
  \begin{lstlisting}[language=proverif,basicstyle=\scriptsize\ttfamily,linewidth=0.85\linewidth]
let rA_2 =
  get tKey(A, B, sk) in
  get tStateASend(st, =A, =B, =sk, na) in
  in(publicChannel, v0: bitstring);
  let (=v2) = g_var_v2_3(v0) in
  let (=na) = g_var_na_4(v0) in
  let kabp = g_var_kabp_5(v0) in
  let (=B) = g_var_B_6(v0) in
  let (=sk) = g_var_sk_7(v0) in
  new stA: bitstring;
  let t = (vA, (vB, (sk, kabp))) in
  event eSecret(A, B, kabp);
  event eCommit_A(A, B, t);
  insert tStateAReceive(stA,A,B,sk,na,kabp);
  out(publicChannel,senc((v3,na),kabp)).
  \end{lstlisting}
\end{lrbox}

\begin{figure*}[t]
  \centering
  \begin{minipage}[t]{0.48\textwidth}
    \centering
    \textbf{Tamarin Rules}\par\vspace{0.5em}
    \usebox{\tamarinbox}
  \end{minipage}%
  \hfill
  \begin{minipage}[t]{0.48\textwidth}
    \centering
    \textbf{ProVerif Translation}\par\vspace{0.5em}
    \usebox{\proverifbox}
  \end{minipage}

  \caption{Translation slice from the Lowe case study. Rule \texttt{A\_2} consumes the server's encrypted response and the initiator's local state, emits the secrecy and agreement events used by the case-study lemmas, records the new local state in a ProVerif table, and sends the third protocol message. Headers, auxiliary destructors, and the rule-instance identifier and DistinctFact instrumentation are omitted for clarity.}
	\label{fig:translation-ex}
  \Description{Side-by-side comparison of a Tamarin rule from Lowe's fix of the Andrew Secure RPC protocol and its ProVerif translation. Left: rule A\_2 consumes key and state facts plus an encrypted input, emits secrecy and agreement events, outputs an encrypted message, and stores state. Right: the corresponding ProVerif process uses table lookups, an input, auxiliary destructors, events, table insertion, and an output.}
\end{figure*}

\Cref{fig:translation-ex} illustrates the translation on Lowe AS rule \texttt{A\_2}, which produces the \texttt{Commit\_A} event used by the non-injective agreement lemma in the case study.
The header, auxiliary destructors, and the instrumentation of \cref{sec:bridging-gaps} are omitted for brevity.
The left-hand side uses Tamarin's concrete syntax for multiset rewriting: the multisets before, between, and after the arrow are the premises, action facts, and conclusions, respectively.
The two \texttt{get} commands implement premise facts stored in ProVerif tables, \texttt{in} implements the network input, the \texttt{event} commands implement Tamarin action facts, and the final \texttt{insert}/\texttt{out} commands implement the conclusion facts.

\subsection{Lemma Translation}
\label{sec:lemma-translation}

Tamarin lemmas are translated into ProVerif queries, all-trace lemmas as correspondence queries and exists-trace lemmas as reachability queries.
Since the two tools use different logics for expressing properties (\cref{fig:trace-property-grammar}), a direct translation is not always possible.
While Tamarin's logic fragment is strictly more expressive than ProVerif's SA-QBF, we follow a systematic rewriting approach that allows us to translate a large class of Tamarin lemmas to logically equivalent ProVerif queries.
Operationally, the translation turns Tamarin action facts $\fact{F}(\vec{t})@i$ into ProVerif events of the same shape and then arranges them into the query form described in \cref{sec:background}: all-trace implications become correspondences, while exists-trace goals become reachability queries.
The Lowe case study illustrates both the direct correspondence translation and a stronger property that requires decomposition.
The verified non-injective agreement lemma for role A has the correspondence shape
\begin{align*}
  \forall A\ B\ t\ i.\ &\factsymbol{CommitA}(A,B,t)@i\\
              \implies &(\exists j.\ \factsymbol{RunningB}(A,B,t)@j \land j<i)\\
              {}\lor{} &(\exists r.\ \factsymbol{Reveal}(A,B)@r).
\end{align*}
The injective variant additionally excludes a distinct second commit for the same payload:
\begin{align*}
  \forall A\ B\ t\ i.\ &\factsymbol{CommitA}(A,B,t)@i \\
              \implies &(\exists j.\ \factsymbol{RunningB}(A,B,t)@j \land j<i\\
                 &\land{} \neg\exists A'\ B'\ i'.\ \factsymbol{CommitA}(A',B',t)@i' \land i'\ne i)\\
              {}\lor{} &R(A,B),
\end{align*}
where $R$ abbreviates the disjunction of key-reveal exceptions.
Although the nested uniqueness condition exceeds SA-QBF, it can be eliminated.
Abstracting from the Lowe event names, let $C_i$ denote a commit, $W_j$ its matching running witness, and $C'_{i'}$ a second commit with the same uniqueness key.
The formula is equivalent to the conjunction of two supported correspondences:
\begin{align*}
 C_i &\implies (\exists j.\ W_j \land j<i) \lor R(A,B),\\
 C_i \land C'_{i'} &\implies i=i' \lor R(A,B).
\end{align*}
For the forward implication, the first query drops uniqueness; applying the original formula to the first commit yields the second query because, absent $R$, its uniqueness conjunct forces the second commit to coincide.
Conversely, absent $R$, the first query supplies a running witness and the second establishes uniqueness of the commit, reconstructing the original conjunctive witness branch.
The escape disjuncts must occur in both queries: omitting them from the uniqueness query would reject duplicate commits that the original property excuses after a reveal.
As in \cref{sec:bridging-gaps}, the second query is emitted using equality of rule-instance identifiers.
\begin{full}
Both components are standard correspondences over the instrumented events, so their soundness follows from the correspondence translation and the exact rule-identifier equality argument of \fullorappendix{rem:instrumented}; the equivalence above transfers this result to the source property.
\end{full}
\begin{conf}
Both components are standard correspondences over the instrumented events.
The second query compares the fresh identifiers assigned to rule instances, so identifier equality enforces the source uniqueness condition; the equivalence above transfers soundness to the source property, with the detailed argument given in \fullorappendix{rem:instrumented}.
\end{conf}

We do not use ProVerif's native \texttt{inj-event}.
It requires an injective mapping from commits to earlier matching runs, so distinct runs can justify duplicated commits with equal parameters.
Tamarin's formula forbids those duplicates unless an escape disjunct holds.
Thus, \texttt{inj-event} is strictly weaker in the verification direction covered by \cref{th:soundness}, whereas the two-query encoding is equivalent to the source formula.
As a validation check, native \texttt{inj-event} agreed with Tamarin in all nine decided cases of a ten-lemma sample; the remaining case timed out.

We first define the root target fragment for our formula rewriting.
\begin{definition}[TSA-QBF]
	\label{def:tsa-qbf}
  Let TSA-QBF be the restriction of Tamarin's logic fragment from \cref{def:tam-logic} to that of ProVerif from \cref{def:sa-qbf}, with $\K(m)@i$ replacing $\PattC(m)@i$.
\end{definition}

Before discussing the rewriting approach, we highlight a key semantic difference between $\K$ facts and ProVerif attacker predicates that influences the supported fragment.
$\Patt{m}@i$ holds whenever $m$ is deducible from the attacker's knowledge at timepoint $i$, whereas $\K(m)@i$ records that the attacker explicitly derived and used $m$.
Hence, $\K(m)@i$ implies $\Patt{m}@i$, and once $\Patt{m}$ holds, it remains true.
\begin{full}
\K{} facts can thus be seen as a subset of $\PattC$ facts, which we make precise in \fullorappendix{def:sr} and illustrate in \cref{ex:k-vs-attacker}.
\end{full}
\begin{conf}
Thus, \K{} facts form a subset of $\PattC$ facts, as captured by the state relation in \fullorappendix{def:sr}.
For example, a rule that receives a pair $\tlst{a, b}$ and logs $\fact{A}(a)$ falsifies $\forall i\ m.\ \fact{A}(m)@i \implies \exists j.\ \K(m)@j$ in Tamarin, while ProVerif proves the $\PattC$ counterpart, as $a$ is deducible from $\tlst{a, b}$.
\end{conf}
Accordingly, the soundness condition is polarity-based: after negating an all-trace property and putting it in negation normal form, every $\K$ atom must occur positively.
The translation renders each accepted $\K$ atom as $\PattC$.
Premise-side existentials are prenexed into ProVerif's universal query variables, while conclusion-only variables receive ProVerif's implicit existential quantification.
A negated $\K$ conclusion is first rewritten as a leading-negation reachability query.
Residual negative knowledge atoms are rejected.
\begin{full}
Thus, the $\K$ case is exactly the positive-atom hypothesis of \fullorappendix{cor:msr-to-pv-satisfaction}.
\end{full}
\begin{conf}
Thus, the $\K$ case is exactly the positive-atom hypothesis required to preserve satisfaction from an MSR trace to the constructed ProVerif execution, as formalized in \fullorappendix{cor:msr-to-pv-satisfaction}.
\end{conf}
$\factsymbol{KU}$ actions are internal to Tamarin's message-deduction strategy and are outside the formal translation; formulas containing them are rejected.
Exists-trace properties containing $\K$ facts are also rejected because \cref{th:completeness} applies only to $\K$-free formulas.

\begin{definition}[T-UKF and P-UKF]
	\label{def:tukf}

	Let T-UKF be the restriction of TSA-QBF to formulas $\varphi$ for which every $\K$ atom in $\operatorname{NNF}(\neg\varphi)$ occurs positively, and P-UKF the corresponding ProVerif logic fragment.
\end{definition}

For our lemma translation, we can further extend the supported fragment in the following way:
If $\varphi$ is a formula in Tamarin's full logic fragment that can be rewritten into $\neg \psi$, $\psi_1 \land \ldots \land \psi_n$, or $\psi_1 \lor \ldots \lor \psi_n$ with $\psi_i \in \text{T-UKF}$, then it can be translated into a negated query or multiple separate queries whose results can be combined to reconstruct the original semantics.
For each generated query, the translation records whether its result is used directly or inverted and emits instructions for recombining multiple queries.
All rewrites are selected from formula structure; protocol, rule, lemma, and event names do not select an encoding.
The rewriting approach is based on two key observations about the differences between the logics.
First, ProVerif does not support disjunctions in reachability queries.
Hence, we split top-level disjunctions in exists-trace lemmas into separate ProVerif queries:
\begin{equation*}
  \begin{gathered}
		(\exists x\ i.\, \fact{A}(x)@i) \lor (\exists y\ j.\, \fact{B}(y)@j) \\
    \text{becomes} \\
    \begin{array}{l}
      \texttt{query x: bitstring, i: time. A(x)@i.}\\
      \texttt{query y: bitstring, j: time. B(y)@j.}
    \end{array}
  \end{gathered}
\end{equation*}
The original semantics are preserved: the lemma holds if and only if 
\emph{at least one} of the resulting queries is satisfied.
The generated output includes the corresponding recombination instruction.
Similarly, ProVerif does not support top-level conjunctions in correspondence queries.
Hence, we split top-level conjunctions in Tamarin formulas for all-trace lemmas into separate ProVerif queries.
The lemma holds if and only if \emph{all} resulting queries are satisfied.
We perform this split only for independently checkable obligations.

The second class of differences arises from ProVerif's lack of support for negation in queries.
We eliminate supported negations by the exact logical equivalences in \cref{tab:rewrite-rules}; the first group converts between all-trace and exists-trace lemmas and records result inversion when a negation is moved to the top level.
The implementation also translates all-trace nonexistence properties by using the leading-negation reachability form.
For the narrow event-free exists-trace case, the implementation instead uses $\PattC(())$ as an always-true premise: the empty tuple is public and always constructible, while ProVerif has no dedicated \texttt{true} query atom.
The generated query then asks whether the forbidden event must occur under this premise and inverts the correspondence result.
For an exists-trace property with a positive witness and a disjunction of possible violations, the translation negates the full property and emits one correspondence from that witness to the complete violation disjunction, again inverting the result.
The disjunction is deliberately not split: every alternative must refer to the same premise witness.
The translation also distributes independent all-trace conjuncts, moves supported guards from a premise to the conclusion, and recognizes an exact two-query conditional-existence decomposition whose no-match query retains the positive matching branch.
These structural cases extend coverage without allowing separately generated queries to choose witnesses independently.
The second table group moves negated subformulas between premises and conclusions of implications.
In particular, the last rule is frequently used to convert secrecy properties expressing that the attacker does not know a secret unless it has been leaked.

\begin{table}[t]
  \caption{Representative exact lemma rewrites. Quantifiers abbreviate closures over all variables in scope. A top-level negation is implemented by inverting the generated query result; a list denotes separately checked queries whose results are conjoined.}
  \label{tab:rewrite-rules}
  \centering
  \small
  \renewcommand{\arraystretch}{1.08}
  \renewcommand{\tabularxcolumn}[1]{m{#1}}
  \begin{tabularx}{\linewidth}{>{\raggedright\arraybackslash}m{.25\linewidth}>{\hsize=.8\hsize\raggedright\arraybackslash}X>{\hsize=1.2\hsize\raggedright\arraybackslash}X}
  \toprule
  \textbf{Transformation} & \textbf{Original} & \textbf{Rewritten} \\
  \midrule
  All-trace to exists-trace &
    $\forall \varphi \implies \neg \exists \psi$ &
    $\neg (\exists \varphi \land \exists \psi)$ \\
  \addlinespace[2pt]
  Exists-trace to all-trace &
    $\exists \varphi \land \neg \exists \psi$ &
    $\neg (\forall \varphi \implies \exists \psi)$ \\
  \addlinespace[2pt]
  Negated $\exists$ to all-trace &
    $\neg (\exists \varphi \land \neg \exists \psi)$ &
    $\forall \varphi \implies \exists \psi$ \\
  \addlinespace[2pt]
  Positive-witness exists-trace &
    $\exists(\varphi\land\bigwedge_k\neg\psi_k)$ &
    $\neg(\forall\varphi\implies\bigvee_k\psi_k)$ \\
  \addlinespace[2pt]
  Independent obligations &
    $\forall\varphi\implies(\psi\land\chi)$ &
    $\{\forall\varphi\implies\psi,\ \forall\varphi\implies\chi\}$ \\
  \midrule
  Move $\neg$ to conclusion & 
    $(\varphi \land \neg \psi) \implies \chi$ & 
    $\varphi \implies (\chi \lor \psi)$ \\
  Move $\neg$ to premise & 
    $\varphi \implies (\neg \psi \lor \chi)$ & 
    $(\varphi \land \psi) \implies \chi$ \\
  \bottomrule
  \end{tabularx}
\end{table}

\subsection{Restriction \& Axiom Translation}
\label{sec:restriction-translation}

Tamarin restrictions are translated to ProVerif restrictions in a straightforward manner similar to lemmas.
However, since ProVerif restrictions also do not support negations and there is no possibility to handle a leading negation, we can only split top-level conjunctions into separate restrictions.
This limits the class of Tamarin restrictions that can be translated.
The implementation does recognize temporal negations of the form $\neg\exists i\ j.\ \fact{A}@i\land\fact{B}@j\land i<j$, rewriting them to the equivalent universal ordering constraint $\fact{A}@i\land\fact{B}@j\implies j<i\lor i=j$.
It also translates a forbidden-event restriction whose opened body is a quantifier-free positive event formula.
For example, $\neg\exists\vec{x}.\,P\land(Q_1\lor Q_2)$ becomes the two prohibitions $P\land Q_1\implies\bot$ and $P\land Q_2\implies\bot$.
The translation computes positive disjunctive normal form and emits one prohibition per branch, with expansion capped at 16 branches.
Restriction conclusions containing nested implications remain unsupported and are omitted with an explicit warning, because ProVerif rejects nested queries in restrictions.

Moreover, there are restrictions that can be translated but require deeper reasoning that is outside the scope of our automatic translation.
For example, the restriction
\begin{equation*}
	\neg \exists x\ i. \fact{Neq}(x, x)@i,
\end{equation*}
which expresses the inequality of two terms, cannot be directly translated to ProVerif due to the negation.
For example, ProVerif can translate an equivalent formulation:
\begin{equation*}
	\forall x\ y\ i. \fact{Neq}(x, y)@i \implies x \neq y
\end{equation*}

Uniqueness restrictions of the form $\forall x\ i\ j.\ \fact{A}(x)@i \land \fact{A}(x)@j \Rightarrow i \doteq j$ need no rewriting, as timepoint equalities in restrictions are directly supported by ProVerif.
If the events carry rule-instance identifiers (\cref{sec:bridging-gaps}), the equality is translated, like any other temporal equality, as an equality of the identifier variables.

Tamarin allows helping lemmas to be reused as axioms in the verification of other lemmas.
We translate helping lemmas to ProVerif axioms.
The limitations that apply to restriction translation also apply to axiom translation, as ProVerif axioms do not support negations either.

The ProVerif build used in the evaluation supports timepoints in restrictions and axioms.

Tamarin supports inductive reasoning for proving lemmas with the \texttt{use_induction} attribute.
This is directly translated to ProVerif's \texttt{[inductive]} query directive.

\subsection{Bridging Semantic Gaps}
\label{sec:bridging-gaps}

\mypara{Multiset Semantics}
One of the fundamental differences between Tamarin and ProVerif is that Tamarin's MSRs operate on multisets of facts, whereas ProVerif's tables are sets without delete operations.
Hence, once a fact is added to a table, it remains there for the rest of the protocol execution.
This leads to an over-approximation of the protocol behavior, as a ProVerif process can use the same fact multiple times, whereas Tamarin can consume it only once.

To avoid this, for each linear storage fact in the right-hand side of a rule, we create a fresh value and add it as the first argument of the corresponding table entry.
For each linear storage fact in the left-hand side of a rule, we then match on the fact together with its fresh value and add a $\fact{DistinctFact}(\mathit{st}, \mathit{args})$ event for each such fact.
We then add a restriction that ensures that no two such events have the same identifier and arguments.
\begin{align*}
	\forall st \, x \, i \, j. \, & \fact{DistinctFact}(st, x)@i \land \\
                              	& \fact{DistinctFact}(st, x)@j \Rightarrow i = j
\end{align*}
The fresh identifier therefore preserves Tamarin's multiset semantics in ProVerif.

\mypara{Simultaneous Action Facts}
Tamarin permits multiple action facts in a single rule that occur at the same timepoint.
In contrast, ProVerif assumes a total order on events.
We can simulate simultaneous action facts by introducing a fresh \emph{rule-instance identifier} that ties all events in the translated process together.
All queries having the same time variable for multiple action facts are rewritten to use distinct time variables for each action fact, with the additional identifier tying them together.
For example, a formula
\begin{align*}
	& \forall x\ i. \; \factsymbol{A}(x)@i \land \factsymbol{B}(x)@i \implies \exists j. \: \factsymbol{C}(x)@j
\shortintertext{is rewritten into}
	& \forall x\ r\ i_1\ i_2. \; \factsymbol{A}(r, x)@i_1 \land \factsymbol{B}(r, x)@i_2 \Rightarrow \exists j. \: \factsymbol{C}(x)@j\,.
\end{align*}
This translation preserves the semantics of the original formula, as the only way for both $\factsymbol{A}(r, x)@i_1$ and $\factsymbol{B}(r, x)@i_2$ to hold is if $r$ is instantiated to the same fresh value.
Sequential emission also exposes a prefix between two events of the same translated rule.
If a final correspondence observes one action fact in its premise and requires a different fact from the same Tamarin action in its conclusion, a shared identifier alone does not show that the second event has already been emitted.
For these cross-fact same-action conclusions, affected rules emit $\fact{RuleCompleted}(r)$ after all ordinary action events, and the translated conclusion is guarded by this completion event.
A same-fact tautology does not require the event, and rules whose action facts are not involved in such a conclusion do not emit it.
The completion event excludes prefixes that stop partway through the translated action.
It is not a source-level timepoint and imposes no progress or liveness restriction.

\mypara{Temporal Equalities}
In Tamarin, two action facts occur at the same timepoint exactly when they belong to the same rule instance; in ProVerif, two distinct events never share a timepoint.
The rule-instance identifiers introduced above make this expressible: we translate a timepoint equality $i \doteq j$ between distinct action facts as an equality between their identifier variables, and a negated equality as a disequality.
Where the equality merely pins a quantified time variable, as in $\exists j.\, \factsymbol{B}@j \land i \doteq j$, we instead eliminate it and reuse the rule-instance-identifier rewriting for simultaneous action facts described above.
Each event occurrence in a query receives its own identifier variable.
Two occurrences share an identifier variable exactly when they are asserted to stem from the same rule instance.
This sharing is derived from the identity of the original source binders: copied occurrences of one source timepoint retain one provenance, whereas independently scoped binders remain distinct even if they have the same printed hint.
Since the identifiers are introduced by the translation and are invisible in the source lemma, this identifier-sharing discipline is the translation's obligation, not the user's.

\mypara{Semantic Differences in the Implication Connective}
Our analysis revealed a semantic distinction in the interpretation of the $\implies$ connective between Tamarin and ProVerif.
While both tools support this connective---Tamarin in all-trace lemmas and ProVerif in correspondence queries---their semantic interpretations differ in a way the translation must account for.

The ProVerif manual \cite[p.55]{proverif-manual} describes non-temporal correspondences such as $\fact{B}\implies\fact{A}$ as requiring the conclusion event $\fact{A}$ to occur before the premise event $\fact{B}$.
The ProVerif build used in the evaluation supports explicit temporal variables.
For formulas with explicit timepoints, however, ProVerif does not insert this conclusion-before-premise order implicitly: a correspondence $\fact{B}@j \implies \fact{A}@i$ constrains $i$ and $j$ only through explicit temporal constraints or through options such as \texttt{conclBeforePremise}.
Without such a constraint, the formula asks only that whenever $\fact{B}$ occurs, some $\fact{A}$ occurs in the same trace.

In contrast, Tamarin employs the standard logical interpretation of $\implies$ without any implicit temporal ordering.
When verifying $\fact{B}@j \implies \fact{A}@i$, Tamarin searches for counterexamples by attempting to satisfy $\fact{B}@j \land \neg \fact{A}@i$---the expected logical negation.
Consequently, in Tamarin, the formula $\fact{B}@j \implies \fact{A}@i$ is satisfied even when event $\fact{B}$ precedes event $\fact{A}$ temporally.
However, this interpretation requires careful handling of trace prefixes.
Without additional constraints, a prefix trace containing only event $\fact{B}$ would constitute a counterexample.
To address this, we must add a \emph{progress} restriction that ensures certain events are reached, thereby preventing incomplete trace prefixes from invalidating the property.
The same prefix issue arises for an unguarded correspondence whose conclusion refers to a later event of the same rule.
Translated cross-fact same-action properties instead use the completion event described above, without adding a progress restriction.
The implication cases used by the translation are covered by the satisfaction relations in the proof; the following experiment only checks the interaction with ProVerif's documented implementation behavior and progress restrictions.

\begin{full}
To validate our understanding of these semantic differences, we conducted a controlled experiment using a minimal model that permits events $\fact{A}$, $\fact{B}$, and $\fact{C}$ to occur only in that sequence. 
We included a progress restriction $\fact{A} \implies \fact{C}$ to ensure that all considered traces terminate with event $\fact{C}$. 
This progress restriction is preserved during translation between the tools.

ProVerif successfully verified all three queries.
In particular, we found that $\fact{A} \implies \fact{B}$, $\fact{B} \implies \fact{C}$, and $\fact{A} \implies \fact{C}$ all hold.
This matches the explicit-timepoint behavior described above: the progress restriction is treated as a liveness condition, filtering traces in which $\fact{A}$ occurs without a corresponding $\fact{C}$, rather than as a precedence constraint.

Thus, the experiment supports the implementation-side assumptions used for progress-restricted examples, while the formal guarantee remains limited to the faithfully translated fragment of \cref{th:soundness}.
\end{full}
\begin{conf}
A controlled minimal model permits events $\fact{A}$, $\fact{B}$, and $\fact{C}$ only in that sequence and preserves $\fact{A}\implies\fact{C}$ as a progress restriction.
ProVerif verifies $\fact{A}\implies\fact{B}$, $\fact{B}\implies\fact{C}$, and $\fact{A}\implies\fact{C}$, confirming the explicit-timepoint behavior.
This supports the implementation-side assumptions used for progress-restricted examples, while the formal guarantee remains limited to the faithfully translated fragment of \cref{th:soundness}.
\end{conf}

\mypara{Relating \K{}-Facts and \textnormal{\texttt{mess}} Facts}
\begin{full}
ProVerif also supports the predicate $\msg{c}{m}$, which we may view as recording that a message $m$ has been communicated on a channel $c$.
Because the action fact $\K(m)$ labels \textsc{MDIn} when $\In(m)$ is produced but not yet consumed, Tamarin's semantics allow visible action facts to occur between these two points in a trace.
We therefore experimented with encoding $\K{}$ facts as $\msg{c}{m}$ facts and with adding restrictions to keep this separation between emission and consumption explicit.
In our experiments, neither approach yielded a practical solution: restrictions caused significant verification overhead, and the encodings did not reliably preserve the intended semantics.
Consequently, this \texttt{mess}-based encoding is not part of the translation and is outside the proof scope.
It is distinct from the supported translation of $\K$ atoms in T-UKF formulas to attacker predicates (\cref{def:tukf}).
\end{full}
\begin{conf}
Experiments with \texttt{mess}-based encodings of \K{} facts did not yield a practical, reliably semantics-preserving solution, so these encodings are outside the translation and proof scope; they are distinct from the supported translation of $\K$ atoms in T-UKF formulas to attacker predicates (\cref{def:tukf}).
\end{conf}

\section{Soundness \& Completeness}
\label{sec:soundness-completeness}

We show that our translation is sound for model/property instances in the faithful translation fragment: if the property of such an instance holds in the ProVerif translation of a Tamarin model, then it also holds in the original Tamarin model.

\begin{definition}[Faithful translation fragment]
	\label{def:faithful-fragment}
	A model/property instance $(\rules,\ET,\varphi)$ is in the faithful translation fragment when the satisfaction-preserving rewritings of \cref{sec:lemma-translation,sec:bridging-gaps} reduce $\varphi$ to a finite conjunction of T-UKF formulas, the model $(\rules,\ET)$ uses only the MSR constructs mapped in \cref{tab:syntax-map} and detailed in \cref{sec:translation}, and its equational theory is classified as faithfully translated in \cref{sec:implementation}.
	Built-ins for which no translation to ProVerif exists, such as Tamarin's \texttt{multiset} and \texttt{bilinear-pairing} built-ins and Diffie-Hellman models requiring inverse, are outside this fragment.
	Models using XOR can be translated, but the best-effort encoding is not part of this fragment.
	Furthermore, each rewritten component is \emph{simultaneity-free}: every time variable indexes exactly one trace atom, and no timepoint equality relates distinct trace atoms.
	Lemmas with simultaneous action facts are supported through the rewriting of \cref{sec:bridging-gaps}, which produces simultaneity-free formulas over identifier-instrumented events.
	\begin{full}
	The detailed argument is given in \fullorappendix{rem:instrumented}.
	\end{full}
	\begin{conf}
	The fresh identifier ties exactly the events emitted by one translated rule instance, so the rewriting preserves which action facts were simultaneous; the detailed argument is given in \fullorappendix{rem:instrumented}.
	\end{conf}
\end{definition}

For a rewritten source property, $\tran{\varphi}$ denotes the conjunction of its translated components.
Soundness applies to each T-UKF component and lifts to the source property through conjunction and the stated satisfaction equivalence.

\begin{restatable}[Soundness]{theorem}{soundness}\label{th:soundness}
	For every $(\rules,\ET,\varphi)$ in the faithful translation fragment:
	\begin{align*}
		\tracespv(\tran{\rules}) \prescript{\forall}{}{\models_\ET^{PV}} \tran{\varphi} \Rightarrow \tracesmsr(\rules) \prescript{\forall}{}{\models_\ET^{T}} \varphi
	\end{align*}
\end{restatable}

For the converse direction, we show that satisfiability is preserved: if a property without $\K$ facts is satisfied by some Tamarin trace, it is satisfied by some trace of the translation.
This covers exists-trace lemmas.
\begin{conf}
\begin{table}[t]
  \caption{Overview of excluded models.}
  \label{tab:discarded-models}

  \small
  \begin{tabularx}{\linewidth}{l@{\hspace{-1em}}rX}
    \toprule
    \textbf{Criterion}                       & \textbf{Count} & \textbf{Reason}            \\ \midrule
    Unsupported built-ins / DH \texttt{inv}  & 327  & \multirow{4}{*}{Out of scope}        \\
    SAPiC models                             & 113  &                                      \\
    Observational equivalence                & 46   &                                      \\
    Accountability lemmas                    & 12   &                                      \\ \midrule
    Invalid DH pattern matching              & 93   & \multirow{2}{*}{Pot.\ unsound} \\
    Well-formedness warnings                 & 35   &                                      \\ \midrule
    Models with stored proofs                & 37   & \multirow{4}{*}{Avoid distortion}    \\
    Duplicate models                         & 27   &                                      \\
    Archived / obsolete models               & 23   &                                      \\
    Translated from other formalisms         & 19   &                                      \\ \midrule
    No lemmas / rules                        & 29   & \multirow{3}{*}{Not verifiable}      \\
    Requires preprocessor flags              & 22   &                                      \\
    Errors / timeout during preprocessing    & 2    &                                      \\ \bottomrule
  \end{tabularx}
\end{table}

\end{conf}
For all-trace lemmas, the converse of \cref{th:soundness} does not hold in general: the translation over-approximates Tamarin's traces---for example, by executions stopping between the events of a single rule---so ProVerif may report spurious counterexamples.
This is consistent with ProVerif's over-approximating design; a ProVerif proof remains conclusive by \cref{th:soundness}.

\begin{restatable}[Completeness for existential properties]{theorem}{completeness}\label{th:completeness}
	For $(\rules,\ET,\varphi)$ in the faithful translation fragment such that $\varphi$ itself is simultaneity-free, belongs to TSA-QBF, and contains no $\K$ facts:
	\begin{equation*}
		\tracesmsr(\rules) \prescript{\exists}{}{\models_\ET^{T}} \varphi \Rightarrow \tracespv(\tran{\rules}) \prescript{\exists}{}{\models_\ET^{PV}} \varphi
	\end{equation*}
\end{restatable}
Both proofs are provided in \fullorappendix{sec:weak-simulation}.

\section{Implementation}
\label{sec:implementation}

The translation is integrated with Tamarin's export framework and builds on the ProVerif export module from \sapicp{}~\cite{sapicplus}.
It supports MSR features including multiset rewrite rules with pattern matching, restrictions, helping lemmas, and the security-property fragment characterized in \cref{def:faithful-fragment}.
The implementation extends Tamarin's export functionality.

\mypara{Built-in Theories}
The implementation distinguishes faithfully translated theories from features that cannot be translated and features that are translated only as best-effort approximations.
Ordinary constructors, user equations accepted by both tools, private helper destructors, and the encodings for MSR state are part of the faithful translation fragment defined in \cref{def:faithful-fragment}.
Tamarin's \texttt{multiset} and \texttt{bilinear-pairing} built-ins cannot be translated.
The same holds for Diffie-Hellman models requiring inverse.
Models using any of these features are excluded from the comparison.
Diffie-Hellman support uses a best-effort translation derived from \sapicp{}, and XOR support is partial.
As discussed in~\cref{sec:eval-outcome}, the XOR encoding is not covered by \cref{th:soundness} because ProVerif cannot faithfully represent associative symbols and the self-inverse XOR equation.
When XOR is used, the implementation emits a warning that the resulting verification results may be unsound.

\mypara{Usage and Warnings}
Translation is invoked via \texttt{tamarin-prover -m=proverif}, producing a self-contained ProVerif model with all necessary destructors, names, tables, and parallel composition of translated rules.
The intended workflow is to keep Tamarin as the source model and use the generated ProVerif file primarily as a fast backend artifact; the file remains inspectable and editable, but its control flow follows the table-based state-machine structure of the source MSRs rather than a handwritten process style.
The implementation provides systematic warnings to guide result interpretation:
When an exact structural normalization produces multiple queries or an inverted result (\cref{sec:lemma-translation}), it records how the query results must be combined.
Colliding bound timepoint names are alpha-renamed automatically before query generation.
When a property remains outside the supported fragment, including any formula containing a $\factsymbol{KU}$ action, the implementation reports the precise rejected shape and emits no approximate query.
If a helping lemma cannot be translated, the implementation reports its omission; this may affect proof search but does not make the trace model incomplete.
If a restriction cannot be translated, the generated model is marked and the user is warned that the verifier outcome concerns an incomplete trace model.

\section{Evaluation}
\label{sec:evaluation}

\begin{full}
\begin{table}[t]
  \caption{Overview of excluded models.}
  \label{tab:discarded-models}

  \small
  \begin{tabularx}{\linewidth}{l@{\hspace{-1em}}rX}
    \toprule
    \textbf{Criterion}                       & \textbf{Count} & \textbf{Reason}            \\ \midrule
    Unsupported built-ins / DH \texttt{inv}  & 327  & \multirow{4}{*}{Out of scope}        \\
    SAPiC models                             & 113  &                                      \\
    Observational equivalence                & 46   &                                      \\
    Accountability lemmas                    & 12   &                                      \\ \midrule
    Invalid DH pattern matching              & 93   & \multirow{2}{*}{Pot.\ unsound} \\
    Well-formedness warnings                 & 35   &                                      \\ \midrule
    Models with stored proofs                & 37   & \multirow{4}{*}{Avoid distortion}    \\
    Duplicate models                         & 27   &                                      \\
    Archived / obsolete models               & 23   &                                      \\
    Translated from other formalisms         & 19   &                                      \\ \midrule
    No lemmas / rules                        & 29   & \multirow{3}{*}{Not verifiable}      \\
    Requires preprocessor flags              & 22   &                                      \\
    Errors / timeout during preprocessing    & 2    &                                      \\ \bottomrule
  \end{tabularx}
\end{table}

\end{full}

\begin{table*}[t]
  \caption{Comparison of Tamarin and ProVerif outcomes on the 523 tasks accepted by the ProVerif front end, using the fixed ProVerif toolchain. Entries marked $^*$ are best-effort XOR cases outside the soundness theorem.}
  \label{tab:comparison-results}
  \centering
  \small
  \begin{tabularx}{0.82\textwidth}{lXrrrrrr}
    \toprule
    \multicolumn{2}{l}{\multirow[c]{2}{*}{\textbf{Tamarin Result}}} & \multicolumn{5}{c}{\textbf{ProVerif Result}}                           & \multirow[c]{2}{*}{\textbf{Total}} \\
    \cmidrule(lr){3-7}
                                                 &         & \textbf{True}    & \textbf{False}  & \textbf{None} & \textbf{OOM} & \textbf{Timeout} & \\
    \midrule
    \multicolumn{7}{l}{\textit{Universal properties (all-trace)}}                                                                                          & 420 \\
    \midrule
                                                 & True    & \textbf{179}     & 1$^\dagger$     & 32            & 60           & 27               & 299 \\
                                                 & False   & 19$^*$           & \textbf{46}     &  4            & 11           &  9               &  89 \\
                                                 & OOM     & 13               & --              &  2            &  9           &  3               &  27 \\
                                                 & Timeout & --               & --              &  2            &  3           & --               &   5 \\
    \midrule
    \multicolumn{7}{l}{\textit{Existential properties (exists-trace)}}                                                                                      & 103 \\
    \midrule
                                                 & True    & \textbf{44}      & 11$^*$           &  8            & 19           & 17               &  99 \\
                                                 & OOM     & --               & --               & --            &  4           & --               &   4 \\
    \midrule
    \multicolumn{7}{l}{\textbf{Total front-end-accepted tasks}}                                                                                             & \textbf{523} \\
    \bottomrule
  \end{tabularx}

  \vspace{0.5em}
  \footnotesize
  OOM = Out of Memory. None denotes a logically inconclusive result; OOM and timeout are resource outcomes.
  Dashes (--) indicate zero cases.\\
  $^*$Due to best-effort XOR translation outside the soundness theorem (see \cref{sec:eval-outcome}).
  $^\dagger$The \texttt{dp3t::upload\_auth} incomplete-model result discussed in \cref{sec:eval-failure-cases}.
\end{table*}

We analyze 121 Tamarin models, examining both where the tools achieve consistent results and where they systematically diverge.
Our evaluation serves two purposes:
\begin{enumerate*}
  \item validating that the translation allows a meaningful comparison, and
  \item comparing the tools' verdicts and resource usage on corresponding tasks.
\end{enumerate*}
\Cref{tab:discarded-models,tab:comparison-results} summarize the dataset construction and cross-tool outcomes.
We first highlight the main findings before examining the dataset, outcome categories, feature ablations, and performance results in~\mbox{detail}.

\subsection{Key Findings}
\label{sec:eval-findings}

Our analysis reveals a substantial common core between the tools.
Among non-XOR tasks, 237 of 238 definitive pairs agree; the remaining pair is explicitly flagged because ProVerif decides an incomplete translation of the \texttt{dp3t} model (\cref{sec:eval-failure-cases}).
Best-effort XOR results remain outside the faithful fragment, with 32 agreements and 30 disagreements among definitive pairs.
The four in-fragment \texttt{functional} lemmas are inconclusive under the fixed ProVerif build.
The primary bottleneck is ProVerif's termination rate on more complex lemmas.
Overall, ProVerif is faster and uses less memory on most comparable tasks.
These results, enabled by our translation, provide the first large-scale empirical comparison of Tamarin and ProVerif on semantically equivalent models.

As described in \cref{sec:lemma-translation}, certain Tamarin lemmas require splitting into multiple ProVerif queries due to top-level disjunctions (in exists-trace lemmas) or conjunctions (in all-trace lemmas), while others require negation handling where the ProVerif result must be negated to obtain the result for the original lemma.
In our evaluation, split-query results are combined automatically according to the recorded connective and negated when required; the generated output carries the same reconstruction instructions.

In several source models, restrictions with event conclusions encode reliable channels, and translating them faithfully requires ProVerif's extended restriction semantics.
We reported a soundness bug in ProVerif and, in ongoing work with the developers, identified a further related one; the developers confirmed both and fixed them in \texttt{working\_dev\_branch} commit \texttt{1d4cd945}, with official releases remaining on \texttt{master}.
The outcome counts use this fixed build.
Relative to the pre-fix ProVerif build used in our initial evaluation, the fix changes only four \texttt{functional} reachability queries (false to inconclusive) and two no-step regression properties (true to inconclusive); no other verdicts change.

\subsection{Dataset \& Scope of Comparison}
\label{sec:eval-dataset}

Our dataset is based on the largest curated set of Tamarin models to date: the example directory in Tamarin's code base.
This diverse collection includes feature showcases, regression tests, and, importantly, the models of various conference publications.
To enable a meaningful comparative analysis, we excluded several categories of models (shown in \cref{tab:discarded-models}) and obtained a representative set of 121 models for our comparison.
The resulting dataset contains 53 publication- or thesis-origin models (40 core published case-study models from major venues and 13 derived feature or test variants), together with 7 classic protocol models, 18 SPORE models, 4 literature benchmarks, and 39 internal or synthetic pedagogical models (121 total).
The models average 208 lines of code overall, while the core published case studies average 381 lines.

This dataset exhibits a survivorship bias: models are typically included only if their analysis (e.g., verification or attack discovery) terminates without user interaction.
Exceptions exist, including legacy models, models created to showcase new features or to highlight bugs, and \enquote{model suites} that use Tamarin's macro language to produce different variants of the same protocol in different threat models.
Some models also come with manually constructed proofs; in such cases, both the original model and a corresponding proof-annotated version---often one per proven lemma---are preserved.

\begin{full}
First, we exclude models that are out of scope for this work.
This includes models using Tamarin built-ins that cannot be faithfully translated (\texttt{multiset}, \texttt{bilinear-pairing}, \texttt{diffie-hellman} with inverse).
Moreover, we discard all SAPiC models, as their translation has already been evaluated in prior work \cite{sapicplus}.
We also exclude models using observational equivalence (diff mode).
Although ProVerif also supports observational equivalence, Tamarin and ProVerif implement it too differently for a direct translation.
Finally, we exclude models focusing on accountability properties, as they require a different treatment.
In total, this includes 498 models.

Second, we exclude models whose soundness of the verification result cannot be guaranteed.
This includes models using ambiguous and therefore invalid pattern matching on DH terms, e.g., $(g^a)^b = (g^b)^a$ cannot uniquely determine whether $x = g^a$ or $x = g^b$.
We also exclude models with critical well-formedness warnings in Tamarin (not subterm convergent equations, missing facts, unbound variables), as we cannot be sure that Tamarin gives the correct result in this case.
This category contains 128 models.

Third, to avoid distortion of the dataset, we exclude models with stored proofs, duplicates of the same model (e.g., used as both feature showcase and regression test), models marked as archived or obsolete, as well as translations from SAPiC+.
This category encompasses 106 models.

Fourth, we discard models without rules or lemmas (used for regression and feature tests), models that require a specific set of preprocessor flags to be evaluated, and models producing errors or timeouts during precomputation.
This encompasses 53 models.
\end{full}
\begin{conf}
As summarized in \cref{tab:discarded-models}, we exclude 498 out-of-scope models using unsupported built-ins or DH inverse, SAPiC, observational equivalence, or accountability properties.
We exclude 128 potentially unsound models because of ambiguous DH pattern matching or critical well-formedness warnings.
To avoid dataset distortion, we exclude 106 models with stored proofs, duplicates, archived or obsolete models, or translations from other formalisms.
Finally, we exclude 53 models that have no rules or lemmas, require particular preprocessor flags, or fail or time out during preprocessing.
\end{conf}

\subsection{The Common Core Logic Fragment}
\label{sec:eval-common-core}

The 121 models contain 566 lemma tasks.
Our translation generates executable queries for 526 tasks.
The ProVerif front end accepts 523 tasks (\qty{92.4}{\percent}).
Of the remaining 43 tasks, 32 formulas contain $\factsymbol{KU}$~actions, which are outside the formal translation, and eight exists-trace formulas contain $\K$ facts, whereas completeness is proved only for $\K$-free formulas.
ProVerif rejects the three generated queries for the Okamoto model because its equational theory is non-convergent, which Tamarin supports~\cite{dreier2017beyond}; the model was designed to demonstrate that extension.
We do not count inconclusive or resource-bound tasks as translation failures.

These cases define the expressiveness boundary.
The translation handles the narrow event-free exists-trace shape, shared positive witnesses, compound all-trace properties, and same-action completion patterns described in \cref{sec:lemma-translation}.

Within the supported fragment, accepted $\K$ atoms become ProVerif \texttt{attacker} predicates.
This covers common secrecy formulas and authentication with reveal when they satisfy the T-UKF polarity condition.
The translation also handles bound-name collisions by alpha-renaming, temporal-negation restrictions, all-trace nonexistence, and the correspondence and uniqueness encoding of injective~agreement.

\begin{table*}[t]
  \caption{Lowe case-study results. In the agreeing-total row, times are summed and memory is the per-lemma peak maximum; injective-agreement measurements use the fixed ProVerif build.}
  \label{tab:lowe-as-case-study}
  \centering
  \small
  \begin{tabularx}{0.82\textwidth}{lX
      S[table-format=2.3]@{\hspace{0.8em}}
      S[table-format=3.1]@{\hspace{1.5em}}
      S[table-format=1.3]@{\hspace{0.8em}}
      S[table-format=2.1]}
    \toprule
    \multirow{2}{*}{\textbf{Lemma}} & \multirow{2}{*}{\textbf{Outcome}} &
    \multicolumn{2}{c}{\textbf{Tamarin}} &
    \multicolumn{2}{c}{\textbf{ProVerif}} \\
    \cmidrule(lr){3-4}\cmidrule(lr){5-6}
    & &
    \multicolumn{1}{c}{\textbf{Time (\unit{\second})}} &
    \multicolumn{1}{c}{\textbf{Memory (\unit{\mebi\byte})}} &
    \multicolumn{1}{c}{\textbf{Time (\unit{\second})}} &
    \multicolumn{1}{c}{\textbf{Memory (\unit{\mebi\byte})}} \\
    \midrule
    \texttt{secrecy}                  & agree & 0.768 &  89.8 & 0.048 & 12.6 \\
    \texttt{noninjectiveagreement_A}  & agree & 0.298 &  81.6 & 0.049 & 12.7 \\
    \texttt{noninjectiveagreement_B}  & agree & 4.076 & 129.9 & 0.048 & 13.1 \\
    \texttt{injectiveagreement_A}     & agree & 2.762 & 157.1 & 0.148 & 14.7 \\
    \texttt{injectiveagreement_B}     & agree & 22.033 & 324.2 & 0.170 & 15.6 \\
    \midrule
    \textbf{Agreeing total} & &
    \bfseries 29.937 &
    \bfseries 324.2 &
    \bfseries 0.463 &
    \bfseries 15.6 \\
    \bottomrule
  \end{tabularx}
\end{table*}

\subsection{Verification Outcome}
\label{sec:eval-outcome}

Having established the expressiveness boundary of both tools, we now examine how they perform.
We separate formally guaranteed results from best-effort XOR results, which are interpreted only as empirical cross-checks.
\Cref{tab:comparison-results} summarizes the verification results categorized by all-trace (correspondence queries) and exists-trace properties (reachability queries).

For all-trace lemmas, Tamarin verifies 299 lemmas, of which ProVerif also verifies 179; ProVerif returns one false result on the incomplete \texttt{dp3t} translation, 32 inconclusive results, 60 out-of-memory~\mbox{results}, and 27 timeouts.
ProVerif also confirms 46 of Tamarin's 89 counterexamples.
In 19 cases, ProVerif verifies a lemma for which Tamarin finds an attack.
These cases are exclusively from models using Tamarin's built-in XOR theory and are expected results of our best-effort translation: the attacker must exploit the self-inverse property of XOR to deduce $x$ from $x \oplus y$ and $y$, which Tamarin's XOR theory handles but ProVerif's approximation cannot.
The other Tamarin-false all-trace lemmas yield 4 inconclusive results, 11 out-of-memory~\mbox{results}, and 9 timeouts in ProVerif.
Since these XOR cases are outside that fragment, they are not counterexamples to \cref{th:soundness}.
The practical reason for including XOR is empirical: many XOR lemmas do not exercise the unsupported associative or self-inverse reasoning.
Across the 29 evaluated XOR models, both tools agree on 32 of the 62 lemmas for which both tools return a Boolean result; the remaining 30 are the XOR disagreements reported in \cref{tab:comparison-results}. For the remaining XOR lemmas, no Boolean cross-tool comparison is available because ProVerif is inconclusive or because at least one tool runs out of memory or times out.
The implementation warning described in \cref{sec:implementation} marks these verdicts as outside the formal guarantee, so they can be used as fast diagnostics but not as standalone soundness evidence.
When Tamarin is resource-bound on all-trace tasks, ProVerif returns true in 13 cases, inconclusive in 4, out-of-memory in 12, and times out in 3; ProVerif has no false results in these rows.

For the 99 exists-trace lemmas that Tamarin verifies, ProVerif agrees in 44 cases and returns false for 11, all of which are best-effort XOR translations; a false reachability result is a proof of unsatisfiability within ProVerif's abstraction, not a failed search.
The remaining runs are 8 inconclusive results, 19 out-of-memory~\mbox{results}, and 17 timeouts; four inconclusive results are the \texttt{functional} lemmas discussed next.
For the four exists-trace lemmas on which Tamarin runs out of memory, ProVerif also runs out of memory.

Across all transformation families, 237 of 238 non-XOR definitive pairs agree.
Among the 45 non-XOR tasks using the exact injective-agreement encoding, all 37 definitive pairs agree; eight are resource-bound.
The only non-XOR disagreement is the flagged \texttt{dp3t} cell.

\subsubsection{Failure Cases and Mitigation Paths}
\label{sec:eval-failure-cases}
Three qualifications explain the cells where the tools do not reach the same definitive result.
First, XOR can change the result because the current encoding does not preserve all XOR reasoning.
For the \texttt{nonce\_secrecy} lemma of \texttt{NSLPK3xor}, Tamarin finds the expected attack, whereas ProVerif verifies the secrecy query under the best-effort approximation.
For the \texttt{executable} lemma of \texttt{CH07}, Tamarin witnesses the expected execution trace, whereas ProVerif proves the corresponding query unreachable under the best-effort encoding.
Future versions could incorporate ProVerif-specific XOR reductions, such as those of \textcite{proverif-xor}, before invoking ProVerif on the translated model; the practical coverage of such reductions on our 30 XOR disagreements has not been quantified in this work.
Second, the four non-XOR exists-trace cases on which ProVerif is inconclusive do not reflect a limitation of the translation.
The \texttt{functional} lemmas in \texttt{mixvote\_ShHh} and \texttt{aletheaDR\_ShHh} (and their \texttt{SmHh} variants) each encode an honest end-to-end vote-counting execution requiring 13 action facts.
Tamarin finds the trace, whereas ProVerif is inconclusive under the fixed ProVerif build.
The unreachable verdicts produced by the pre-fix ProVerif build result from the locally reproduced restriction-handling defect; all all-trace results in the affected theories are unchanged under the fixed build.
Third, \texttt{dp3t::upload\_auth} is false in ProVerif but true in Tamarin only after four time-ordering restrictions with nested implications are dropped, so this ProVerif result concerns an incomplete, weaker model rather than a faithful disagreement; \texttt{dp3t::soundness} is separately excluded because its formula is not translated.
The same unsupported restriction shape occurs in \texttt{robert}, but its three tasks are now among the 32 rejected because their formulas contain $\factsymbol{KU}$~actions.

\subsection{Impact of Translation Features}
\label{sec:eval-encoding}

\begin{figure*}[t]
    \centering
    \includegraphics[width=\textwidth]{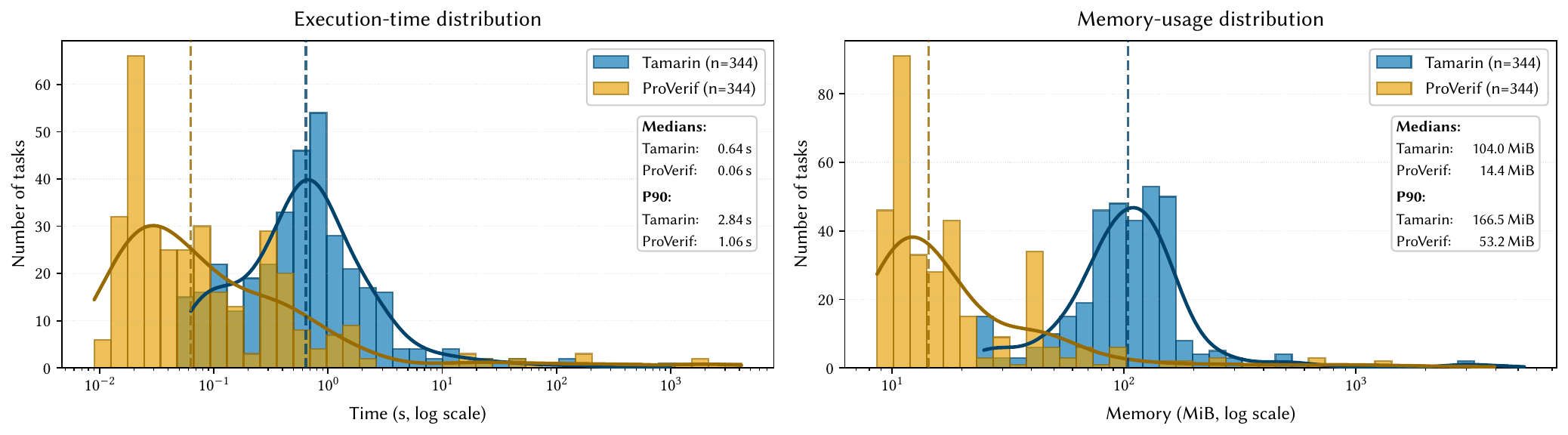}

    \caption{Performance on the $n=344$ tasks where Tamarin returns a Boolean result and ProVerif completes with a logical result.
    (Left) Execution time: ProVerif median \qty{0.063}{\second} and P90 \qty{1.06}{\second} vs.\ Tamarin median \qty{0.640}{\second} and P90 \qty{2.84}{\second}.
    (Right) Peak memory: ProVerif median \qty{14.4}{\mebi\byte} and P90 \qty{53.2}{\mebi\byte} vs.\ Tamarin median \qty{104.0}{\mebi\byte} and P90 \qty{166.5}{\mebi\byte}.}
    \label{fig:performance}
    \Description{Two distributions compare Tamarin and ProVerif on 344 completed task pairs. The execution-time and memory plots both place ProVerif at lower values for most tasks.}
\end{figure*}

We characterize the impact of several key features of both ProVerif and our translation.
For each feature comparison, we rerun all 566 tasks and assess verdict preservation on the tasks for which both evaluation settings return Boolean results.
First, we are interested in the effectiveness of our formula rewriting and encoding techniques for multiset rewriting and simultaneous events introduced in \cref{sec:translation}.
Second, we analyze the impact of ProVerif's precision settings as described in \cite{proverif-manual}.
Third, we investigate whether ProVerif benefits from helping lemmas written specifically for Tamarin.

\mypara{Encoding Techniques}
Among the 523 front-end-accepted tasks (420 all-trace, 103 exists-trace), 78 all-trace lemmas produce multiple ProVerif queries, including the correspondence-plus-uniqueness encoding of injective~agreement; no exists-trace lemma produces multiple queries.
The reconstructed result is inverted for 28 all-trace and 73 exists-trace lemmas.
Source-timepoint splitting is used in 45 all-trace and 13 exists-trace lemmas.
This shows that our rewriting and encoding techniques are widely applicable in practice.
To assess the impact of our multiset encoding, we compare the default evaluation setting with an otherwise identical run in which this encoding is disabled.
The two evaluation settings agree on 236 of 256 paired Boolean results.
In all 20 disagreements, Tamarin and ProVerif with the default setting verify a complete, non-XOR all-trace property, whereas ProVerif without the multiset encoding reports a counterexample.
These spurious counterexamples show that the encoding is necessary to enforce the linear consumption of facts and preserve the source MSR semantics.

\mypara{Precision in ProVerif}
By default, ProVerif's Horn clause abstraction ignores the number of repetitions of actions, potentially allowing multiple executions of constructs that should execute only once---for example, an input $\Pin(c, x)$ can be reused multiple times, leading to false attacks \cite[pp. 91]{proverif-manual}.
ProVerif provides the directive \texttt{set preciseActions = true} that ensures each action executes at most once per process instance \cite{CCT18}.
For translations from MSRs, where each $\In$ fact is inherently consumed exactly once, enabling this setting is obligatory to maintain semantic equivalence.

Disabling precision has asymmetric soundness implications: for all-trace lemmas, verified results remain sound (the overapproximation still holds), but for exists-trace lemmas, only \emph{disproven} results are reliable---verified traces may rely on input reuse not possible in Tamarin.
To assess the practical impact, we compare the default evaluation setting with an otherwise identical run that omits \texttt{preciseActions}.
The two evaluation settings agree on 206 of 208 paired Boolean results.
In both disagreements, Tamarin and ProVerif with \texttt{preciseActions} verify the property, whereas ProVerif without this directive reports a counterexample: \texttt{injective\_agree} in \texttt{NSLPK3} and \texttt{injectiveagreement\_A} in \texttt{Otway-Rees}.
Moreover, 105 Boolean results obtained with the default setting become inconclusive without precision.
Thus, \texttt{preciseActions} is required to prevent spurious input reuse and preserve the source MSR semantics, rather than being merely a proof-search optimization.

\mypara{Incomplete Models}
Of the 523 front-end-accepted tasks, 455 use complete trace models (366 all-trace and 89 exists-trace), while 68 omit at least one restriction (54 all-trace and 14 exists-trace).
Separately, 219 omit at least one helping lemma (184 all-trace and 35 exists-trace); the restriction and axiom categories overlap on 29 tasks.
Helping lemmas are translated to axioms and do not change process traces, so their omission is not a model-completeness failure, although it may affect proof search and termination.
Omitting a restriction can only enlarge the trace set.
Consequently, verified all-trace results and disproven exists-trace results remain conservative, whereas disproven all-trace results and verified exists-trace results may depend on the incomplete model.
Among the 68 accepted translations with a missing restriction, ProVerif disproves exactly one all-trace lemma and verifies no exists-trace lemma.
The sole non-conservative logical verdict is \texttt{dp3t::upload\_auth}, the only definitive non-XOR disagreement; we exclude it from the faithful comparison set.

\mypara{Helping Lemmas}
\begin{full}
We investigated whether ProVerif benefits from helping lemmas written for Tamarin, as both tools use very different reasoning methods.
We compare the default evaluation setting with a run that omits translated helping lemmas.
The two evaluation settings agree on all 313 paired Boolean results.
Among resource-bound results, 18 out-of-memory~\mbox{results} become timeouts and six timeouts become out-of-memory~\mbox{results}; no task crosses between a logical and resource-bound outcome.
\end{full}
\begin{conf}
Across all 313 paired Boolean results, omitting translated Tamarin helping lemmas changed no verdict, and no task crossed between logical and resource-bound outcomes.
The only shifts are among resource-bound results, where 18 out-of-memory~\mbox{results} become timeouts and six timeouts become out-of-memory~\mbox{results}.
\end{conf}

Overall, these feature comparisons show that the multiset encoding and \texttt{preciseActions} are required for Boolean verdicts to soundly reflect the source MSR model, whereas Tamarin-specific helping lemmas do not improve ProVerif's Boolean results or help it terminate.

\subsection{Performance Characterization}
\label{sec:eval-performance}

Each lemma and query was verified independently in a containerized environment.
All tasks were evaluated with a timeout of at most \qty{7200}{\second} and a memory limit of at most \qty{96}{\gibi\byte}.
Tamarin used 12 configured CPUs per task.
ProVerif is single-threaded and therefore used one CPU per task.
We measured the total wall clock running time and peak memory usage during verification.
In a separate translation-overhead measurement, translating all 121 models took about \qty{72}{\second} in total, with a median translation time of \qty{0.22}{\second} per model.
\Cref{fig:performance} shows a comparison between Tamarin's and ProVerif's performance across the evaluation.

\begin{full}
The Lowe case study demonstrates the practical value behind these aggregate numbers.
The model has 206 lines and 7 rules, comes from the SPORE protocol library, and captures Lowe's fix of the Andrew Secure RPC protocol.
A systematic scan found 237 complete, agreeing, Boolean, non-XOR tasks; Lowe is a representative published case study with security-meaningful agreement lemmas.
As shown in \cref{tab:lowe-as-case-study}, ProVerif is $\num{64.7}\times$ faster across all five agreeing all-trace lemmas, well above the $\num{6.74}\times$ median per-task speedup reported below.
The correspondence-plus-uniqueness encoding also translates both injective-agreement lemmas, and the fixed ProVerif build verifies both in agreement with Tamarin.
The \texttt{injectiveagreement\_B} row is $\num{129.6}\times$ faster and saves \qty{21.863}{\second}.
Across the five displayed lemmas, the recorded totals are \qty{0.463}{\second} for ProVerif and \qty{29.937}{\second} for Tamarin.
The exists-trace lemma is excluded from that total because Tamarin returns true while ProVerif is logically inconclusive.
\end{full}
\begin{conf}
The published Lowe case study from the SPORE library, which captures Lowe's fix of Andrew Secure RPC, is representative of 237 complete, agreeing, Boolean, non-XOR tasks.
One lemma is $\num{129.6}\times$ faster and saves \qty{21.863}{\second}.
Across its five agreeing all-trace lemmas in \cref{tab:lowe-as-case-study}, including both injective-agreement lemmas, ProVerif takes \qty{0.463}{\second} versus Tamarin's \qty{29.937}{\second}, a $\num{64.7}\times$ speedup.
The exists-trace lemma is excluded because ProVerif is logically inconclusive.
\end{conf}

Among the 344 tasks in \cref{fig:performance}, ProVerif is faster in 316 cases (\qty{91.9}{\percent}).
The median per-task time ratio is $\num{6.74}\times$, and the median per-task memory ratio is $\num{6.13}\times$.
The marginal median runtime is \qty{0.063}{\second} for ProVerif versus \qty{0.640}{\second} for Tamarin; for peak memory, it is \qty{14.4}{\mebi\byte} versus \qty{104.0}{\mebi\byte}.
The corresponding runtime P90 values are \qty{1.06}{\second} versus \qty{2.84}{\second}; the memory P90 values are \qty{53.2}{\mebi\byte} versus \qty{166.5}{\mebi\byte}.
Of the 28 Tamarin-faster tasks, 21 exercise its auto-source feature~\cite{auto-sources}; the other seven include injective and reachability queries, so the advantage is workload-dependent rather than confined to one feature family.
ProVerif's consistently lower resource needs make it well-suited for rapid prototyping.

\FloatBarrier

\section{Conclusion}
\label{sec:conclusion}

In this work, we presented a sound translation from Tamarin to ProVerif that enables systematic comparative analysis of these prominent verification tools.
With our techniques for formula rewriting, encoding multiset rewrite semantics, and handling simultaneous events, we bridged key formalism gaps,
revealing both the similarities and fundamental differences between
the tools.

Our comparative analysis identified a substantial common core logic fragment.
Across 121 models, 523 of 566 lemma tasks (\qty{92.4}{\percent}) have executable translations accepted by the ProVerif front end.
The tools agree on 237 of 238 definitive non-XOR pairs; the sole exception is the flagged \texttt{dp3t} verdict on an incomplete model.
Best-effort XOR is separate, with 32 agreements among 62 definitive pairs.
ProVerif returns true for 13 all-trace tasks on which Tamarin runs out of memory; these resource-bound Tamarin cases are not counted as definitive cross-tool comparisons.
These results expose the trade-off between ProVerif's speed and Tamarin's expressiveness and provide a methodology for combining backends within explicit fragment boundaries.

\begin{acks}

We thank Julian Biehl for the foundational rule-translation design and implementation in his Master's thesis~\cite{msc-biehl}.

This paper was edited for grammar using ChatGPT.
\end{acks}

\printbibliography

\appendix
\crefalias{section}{appendix}
\crefalias{subsection}{appendix}
\crefalias{subsubsection}{appendix}

\section{Open Science}
\label{sec:open-science}

\begin{conf}
The companion artifact provides the complete 566-task inventory and primary-run results, including outcomes, execution times, and memory usage for both tools.
The artifact is available at:
\begin{center}
  \url{https://doi.org/10.5281/zenodo.21827126}
\end{center}
It includes scripts for running tasks, merging results, reconstructing the primary-run statistics, and regenerating~\cref{fig:performance}, together with a Docker recipe that builds the evaluation image from public sources using the exact ProVerif revision used in our experiments.
The included \texttt{README.md} documents the data formats, evaluation procedure, and analysis commands.
The supplied JSONL file directly reconstructs the primary outcome and performance statistics, while the Docker recipe and the evaluation script's options support reproducing the feature-ablation results by rerunning the tasks.
The implementation and evaluation scripts are GPL-3.0 licensed, and the evaluation data are released under CC0-1.0.
\end{conf}

\begin{full}
We provide a companion artifact containing the primary evaluation data, analysis and runner scripts, and a Docker recipe for building the evaluation image from public sources.
The artifact is available at:
\begin{center}
  \url{https://doi.org/10.5281/zenodo.21827126}
\end{center}

\mypara{Artifacts Provided}
Our artifact includes:
\begin{enumerate*}[label=(\arabic*)]
  \item the 566-task evaluation inventory extracted from Tamarin's protocol examples;
  \item complete primary-run evaluation results with outcomes, execution times, and memory usage for both tools;
  \item Python and shell scripts for running tasks, merging results, reconstructing the primary-run statistics, and regenerating the performance figure; and
  \item a Docker recipe that obtains Tamarin from its public \texttt{develop} branch and builds the exact ProVerif revision used in our experiments.
\end{enumerate*}

\mypara{Reproducibility}
The included \texttt{README.md} documents the dataset and result formats, the evaluation procedure, and the analysis commands.
The primary outcome and performance claims can be reconstructed immediately from the supplied JSONL file.
The supplied Docker recipe builds the evaluation environment, and the evaluation script's options enable the feature-ablation results to be reproduced by rerunning the tasks from public sources.

\mypara{Licensing}
The implementation and evaluation scripts are GPL-3.0 licensed.
The evaluation inventory and result data are released under CC0-1.0.
\end{full}

\begin{full}
\section{Soundness and Completeness Proofs}
\subsection{Auxiliary Definitions}
\label{sec:app-aux-defs}

\newcommand{\pv}{\mathit{pv}}

Let $\pvconf$ be a ProVerif configuration and \textsc{Red} be one of ProVerif's reduction rules (\cite[Figure~8]{sapicplus-full}).
We denote the transition from configuration $\pvconf$ to configuration $\pvconf'$ with the reduction rule $\textsc{Red}$ and label $l$ by
\begin{equation*}
	\pvconf \withpv{l}_{\textsc{Red}} \pvconf',
\end{equation*}
where $l$ may be omitted depending on the rule.
Similarly, for MSR states, we write $\mss \withmsr{a}_{\ri} \mss'$ for a rewriting step from $\mss$ to $\mss'$ with the ground rule instance $\ri = \ru\sigma \in \ginsts_\ET(\allrules)$ of a rule $\ru$ under a grounding substitution $\sigma$.

In Tamarin, the attacker knowledge consists of the terms $m$,
for which there are $!\K(m)$ facts in the current MSR state.
During execution, the attacker can raise a $\K(m)$ fact to show that it has explicit knowledge of the term $m$,
i.e., that it has deduced it at some previous timepoint.
We can use such facts in trace formulas to ask if the attacker has raised such an event at a given time.
In ProVerif, the attacker knowledge consists of the terms contained in the set $\Att$,
which is part of the process configuration.
We can ask if the attacker has knowledge of a term $m$ through $\Patt{m}$,
but, as mentioned before, $\Att$ also contains messages that are deducible.

\begin{definition}[Attacker knowledge in ProVerif]
	The adversary \textit{can deduce} $m$ from a process configuration $\pvconf$ as defined in \cite[Sec.~2.3.1]{proverifoverhaul} if
	there exists a configuration
$\pvconf' = (\E', \Pro, \Tbl, \Att')$
with $m \in \Att'$ reachable from $\pvconf$ by only applying \textsc{App} and \textsc{New}.
\end{definition}
Based on this notion, we define two functions: $\ded(\cdot)$ and $\kn(\cdot)$.
For a ProVerif configuration $\pvconf$, $\ded(\pvconf)$ returns all the terms deducible by the adversary from $\pvconf$.
For an MSR state $\mss$, $\kn(\mss) = \set{ m \given !\K(m) \in \mss }$ returns the set of all messages known to the adversary in the MSR state $\mss$.

Let $\ri \in \ginsts_\ET(\rules)$ be a ground instance of an MSR.
We define the set of ProVerif table entries $\inserts(\ri)$ as follows
\begin{equation*}
	\set*{ \factsymbol{F}_{tbl}(t_1, \ldots, t_n) \given \factsymbol{F}(t_1, \ldots, t_n) \in \concls(\ri)|_{\Storage}}\,,
\end{equation*}
where $\factsymbol{F}_{tbl}$ is a ProVerif table corresponding to the fact symbol $\factsymbol{F}$.

After the successful execution of the process $\tran{\ri}$,
$\inserts(\ri) \subseteq \Tbl$ holds for the resulting configuration,
which corresponds to the MSR state after applying $\ri$.

\begin{definition}[Applicability of a translated process]
	\label{def:tran-applicable}

	Let $\pvconf$ be a ProVerif configuration and $r \in \rules$ be a rule with grounding substitution $\sigma$.
	Then $\tran{r \sigma}$ is applicable in $\pvconf$ if all sub-processes of the premise (the first two lines in \cref{def:tran-rule}) of $\tran{r \sigma}$ can be executed starting from $\pvconf$:
	\begin{enumerate}[label=(\arabic*)]
		\item All ground instances of the premises have a corresponding table entry.
		\item The attacker has explicit knowledge of all terms it needs to input to the process.
	\end{enumerate}
\end{definition}

\subsection{Satisfaction Relation}
\label{sec:app-satisfaction}

We introduce a satisfaction $\models_\ET$ w.r.t. an equational theory $\ET$.
For a trace $tr$, let $\idx(tr) \defas \set{1,\ldots,|tr|}$ denote its index set.

\begin{definition}[Valuation]
  Let $\mathbf{D}_s$ be the domain associated to sort $s$; we set $\mathbf{D}_{\msgsort}$ to the set of ground terms and $\mathbf{D}_{\tempsort}$ to the set of trace indices $\mathbb{N}$.
  A function $\theta$ from $\Vars$ to $\mathbf{D}_{\msgsort} \cup \mathbf{D}_{\tempsort}$
  is called a valuation if it respects sorts, i.e., $\theta(\mathcal{V}_s) \subseteq \mathbf{D}_s$ for all sorts $s$.
  For a term $t$, $t\theta$ denotes the application of the homomorphic extension of $\theta$ to $t$
  where variables outside the domain of $\theta$ remain unchanged.
\end{definition}

\begin{definition}[Satisfaction relation]
	\label{def:satisfaction-rel}
	We define the satisfaction relation $(tr, \theta) \models_\ET \varphi$ between a trace $tr$, a valuation $\theta$, and a formula $\varphi$ inductively as follows:
	the remaining connectives ($\lor$, $\Rightarrow$, $\forall$) and the derived comparisons $i \dotleq j$ (i.e., $i \dotless j \lor i \doteq j$) and $i \not\doteq j$ are defined as usual.
	\begin{align*}
		(\tr, \theta) & \models_\ET \bot                               &  & \text{never}                                                                                                                                                        \\
		(\tr, \theta) & \models_\ET \factsymbol{F}@i                   &  & \text{if $\theta(i) \in \idx(\tr)$ and $\factsymbol{F}\theta \in_\ET \tr_{\theta(i)}$}                                                                     \\
		(\tr, \theta) & \models_\ET i \dotless j                       &  & \text{if $\theta(i) < \theta(j)$}                                                                                                                                \\
		(\tr, \theta) & \models_\ET i \doteq j                         &  & \text{if $\theta(i) = \theta(j)$}                                                                                                                                \\
		(\tr, \theta) & \models_\ET t_1 = t_2                          &  & \text{if $t_1\theta =_\ET t_2\theta$}                                                                                                                                \\
		(\tr, \theta) & \models_\ET \neg{\varphi}                      &  & \text{if not $(\tr, \theta) \models_{\ET} \varphi$}                                                                                                                       \\
		(\tr, \theta) & \models_\ET \varphi \land \psi                &  & \text{if $(\tr, \theta) \models_{\ET} \varphi$ and $(\tr, \theta) \models_{\ET} \psi$}                                                                                       \\
		(\tr, \theta) & \models_\ET \exists \vec{x} : \vec{s}. \varphi &  & \begin{aligned}[t]
			\text{if } & \exists \vec{u} \in \mathbf{D}_{s_1} \times \cdots \times \mathbf{D}_{s_{|\vec x|}} . \\
			           & (\tr, \theta\left[\vec{x} \mapsto \vec{u}\right]) \models_{\ET} \varphi
		\end{aligned}
	\end{align*}
\end{definition}

For ProVerif, we extend the above definition to executions of ProVerif processes.
For all but the following additional case, we project the execution to the corresponding trace.
\begin{definition}[Satisfaction relation for ProVerif]
	The satisfaction relation $\satpv$ is defined by extending \cref{def:satisfaction-rel} to ProVerif executions.
	For an execution $e = \pvconf[0] \withpv{l_1} \pvconf[1] \cdots \withpv{l_n} \pvconf[n]$, we add the following case:
	\begin{equation*}
		\begin{aligned}
			(e, \theta) \satpv \Patt{m}@i \quad \text{if }&
			\theta(i) \in \{1,\ldots,n\}\text{ and}\\
			&m\theta \in \ded(\pvconf[\theta(i)]).
		\end{aligned}
	\end{equation*}
\end{definition}

Similarly to \textcite{tamarin12}, for closed formulas $\varphi$,
we overload notation and write $tr \models_\ET \varphi$ if $(tr, \theta) \models_\ET \varphi$ for all $\theta$.
In the following, we refer to the unfiltered traces as \emph{strong} traces, in contrast to the weak traces of \cref{def:traces}.
Next, we recall a lemma from~\cite[Appendix C, Section E, Lemma 19]{tamarin12}.
\begin{lemma}
	\label{lem:msrwtrtr}
	For all equational theories $\ET$, strong MSR traces $tr_{msr}$, corresponding weak MSR traces $wtr_{msr}$ and closed formulas $\varphi$:
	\[
		tr_{msr} \sattam \varphi \Leftrightarrow wtr_{msr} \sattam \varphi
	\]
\end{lemma}
We also adapt the above lemma to ProVerif. The statement follows by the same proof.
\begin{lemma}
	\label{lem:pvwtrtr}
	For all equational theories $\ET$, strong ProVerif traces $tr_{pv}$, corresponding weak ProVerif traces $wtr_{pv}$ and closed formulas $\varphi$ that do not contain $\PattC$ facts:
	\[
		tr_{pv} \satpv \varphi \Leftrightarrow wtr_{pv} \satpv \varphi
	\]
\end{lemma}
The condition that $\varphi$ does not contain $\PattC$ facts can be relaxed,
but this form of the statement is sufficient for our purposes.

\begin{definition}[Silencing facts in traces]
	Let $tr$ be an MSR or ProVerif trace and $\fact{F}$ a fact symbol.
	We define the trace $tr_{\neq \fact{F}}$ as the trace $tr$ in which the facts with symbol $\fact{F}$ have been made silent,
	i.e., every such fact is removed from the multiset at each index of $tr$.
\end{definition}

\begin{lemma}
	\label{lem:msrktr}
	For all equational theories $\ET$, Tamarin traces $tr$, corresponding Tamarin traces $tr_{\neq \K}$ and closed formulas $\varphi$ that do not contain $\K$ facts:
	\[
		tr \sattam \varphi \Leftrightarrow tr_{\neq \K} \sattam \varphi
	\]
\end{lemma}
To obtain a proof for the lemma, one needs to modify the proof of~\cref{lem:msrwtrtr}
by performing induction on the number of $\K$ transitions in $tr$,
instead of the number of silent actions in $tr$.

We can use the above definition and a slightly adapted proof of the previously cited lemma
to obtain the following result for formulas in SA-QBF
\begin{lemma}
	\label{lem:pvmesstr}
	For all equational theories $\ET$, ProVerif traces $tr$, corresponding ProVerif traces $tr_{\neq \msgC}$ and closed formulas $\varphi \in \text{SA-QBF}$:
	\[
		tr \satpv \varphi \Leftrightarrow tr_{\neq \msgC} \satpv \varphi
	\]
\end{lemma}
To obtain a proof for the lemma, one needs to modify the proof of~\cref{lem:msrwtrtr}
by performing induction on the number of $\msgC$ transitions in $tr$,
instead of the number of silent actions in $tr$.

\begin{definition}[Satisfiability]
	Let $\mathit{Tr}$ be a set of traces and $\models_\ET$ a satisfaction relation.
	A formula $\varphi$ is \emph{satisfiable} for $\mathit{Tr}$, written $\mathit{Tr} \prescript{\exists}{}{\models_\ET} \varphi$, if
	\[
		\exists \, \tr \in \mathit{Tr},\: \theta. \; (\tr , \theta) \models_\ET \varphi\,.
	\]
\end{definition}
\begin{definition}[Validity]
	Let $\mathit{Tr}$ be a set of traces and $\models_\ET$ a satisfaction relation.
	A formula $\varphi$ is \emph{valid} for $\mathit{Tr}$, written $\mathit{Tr} \prescript{\forall}{}{\models_\ET} \varphi$, if
	\[
		\forall \, \tr \in \mathit{Tr},\; \theta. \; (\tr , \theta) \models_\ET \varphi\,.
	\]
\end{definition}

\subsection{Proof Overview}
\label{sec:prfoverview}

The proof establishes a weak simulation between the LTS induced by the MSR system $\rules$ and its ProVerif translation $\tran{\rules}$, where
\begin{enumerate*}[label=(\arabic*)]
  \item initial states are related, and
  \item each visible MSR transition corresponds to ProVerif transitions producing a lexicographically ordered trace.
\end{enumerate*}
Our trace comparison ignores Tamarin's $\K$ facts and ProVerif's $\msgC$ facts: the former are translated to $\PattC$ and compared via deducibility rather than trace inclusion, while the latter do not affect SA-QBF satisfaction (\cref{lem:pvmesstr}).
The weak simulation maintains a state equivalence relating Tamarin's $!\K$ facts to ProVerif's attacker knowledge, enabling soundness for T-UKF and the preservation of exists-trace properties for $\K$-free TSA-QBF formulas.

\subsection{Trace Comparison}
\label{subsec:trace-comparison}

To simplify the comparison between Tamarin and ProVerif traces,
we lift ProVerif traces to sequences of singleton sets of facts.
For Tamarin traces, we consider them as sequences of sets of facts since the multiset semantics cannot impact formula satisfaction.

Overloading notation, we write $\lex(tr_{msr})$ to denote the lifting of a Tamarin trace $tr_{msr}$ to a sequence of singleton sets of facts ordered lexicographically.
For example, with $tr_{msr} = (\{A, B\}, \{C\})$, we have $\lex(tr_{msr}) = (\{A\}, \{B\}, \{C\})$.
We also lift the function to sets of traces.

\begin{lemma}[Valuation Translation]
	\label{lem:val-translation}
  Let $tr = (a_1, \ldots, a_n)$ be a trace and $\varphi$ be a simultaneity-free formula (\cref{def:faithful-fragment}), i.e., each time variable of $\varphi$ indexes exactly one trace atom and no timepoint equality relates distinct trace atoms.
  The formulas produced by the lemma translation satisfy these conditions after the rewriting of \cref{sec:bridging-gaps} (see \cref{rem:instrumented}).
  Let $\sigma$ be a valuation that satisfies $\varphi$ on trace $tr$.
	Then there exists a valuation $\sigma'$ that satisfies $\varphi$ on trace $\lex(tr)$ defined as follows:
  
  For each time variable $i$ where $\fact{F}@i$ appears in $\varphi$ and $\sigma(i) = k$:
  \begin{equation*}
    \sigma'(i) = \sum_{\ell=1}^{k-1} |a_\ell| + \text{pos}_{\text{lex}}(\fact{F}, a_k)
	\end{equation*}
  where $\text{pos}_{\text{lex}}(\fact{F}, a_k)$ denotes the position of $\fact{F}$ in the lexicographic ordering of $a_k$.
\end{lemma}
\begin{proof}
  For any $\fact{F}@i$ in $\varphi$ with $\sigma(i) = k$ we have $\fact{F} \in a_k$ by definition of $\sigma$.
  By the definition of $\sigma'$ this places $\fact{F}$ at position $\sigma'(i)$ in $\lex(tr)$.

  For any constraint $i < j$ in $\varphi$ where $\sigma(i) = k_1$ and $\sigma(j) = k_2$:
  \begin{itemize}
      \item Since $\sigma$ satisfies $i < j$, we have $k_1 < k_2$
      \item $\sigma'(i) \leq \sum_{\ell=1}^{k_1} |a_\ell|$
      \item $\sigma'(j) \geq \sum_{\ell=1}^{k_2-1} |a_\ell| + 1$
      \item Since $k_1 < k_2$, we have $\sigma'(i) < \sigma'(j)$
  \end{itemize}
  Timepoint disequalities $i \neq j$ are preserved since satisfying them requires $k_1 \neq k_2$, and timepoint equalities relate only occurrences of the same atom, which are mapped to the same position.
Therefore, $\sigma'$ preserves all action occurrences and temporal orderings, ensuring $\varphi$ is satisfied on $\lex(tr)$.
\end{proof}

\begin{corollary}
	\label{lem:lex-satisfaction}

	Let $\mathit{Tr}$ be a set of traces and $\varphi$ be a closed formula satisfying the conditions of \cref{lem:val-translation}.
	If $\mathit{Tr} \prescript{\forall}{}{\models_{\ET}} \varphi$, then $\lex(\mathit{Tr}) \prescript{\forall}{}{\models_{\ET}} \varphi$.
\end{corollary}
\begin{proof}
	Every trace in $\lex(\mathit{Tr})$ is $\lex(tr)$ for some $tr \in \mathit{Tr}$.
	The claim follows from \cref{lem:val-translation}.
\end{proof}

\subsection{State Relation}
\label{subsec:staterel}

Each step of the simulation needs to assume the two models are in a relation in order to prove that this relation is retained.
We thus define
a relation between MSR states (multisets) and ProVerif states (process configurations).

\begin{definition}[State relation]
	\label{def:sr}
	Let $\pvconf = \pvconfquad$ be a process configuration and $\mss$ be an MSR state.
	We relate $\mss$ to $\pvconf$, i.e., $\mss \sim \pvconf$, if the following conditions hold:
	\begin{enumerate}[label=(\arabic*)]
		\item $\mss|_{\Storage} \subseteq_{tbl} \Tbl$
		\item $\kn(\mss) \subseteq \ded(\pvconf)$
	\end{enumerate}
\end{definition}
In the above definition, $\mss|_{\Storage} \subseteq_{tbl} \Tbl$ means that all grounded \Storage facts in the MSR state \mss have a corresponding table entry stored in \Tbl.
In particular, $\mss|_{\Storage} \subseteq_{tbl} \Tbl$ if and only if, for each \Storage fact $\fact{F}(t_1, \ldots, t_n) \in \mss$, we have a matching entry $\fact{F}_{tbl}(t_1, \ldots, t_n) \in \Tbl$, where the table name $\fact{F}_{tbl}$ corresponds to the fact symbol $\fact{F}$.
These conditions help us in two ways.
The first condition helps us show
that if a rule instance $\ri \in \ginsts_\ET(\rules)$ is applicable to $\mss$,
then $\tran{\ri}$ is applicable to $\pvconf$.
Moreover, it guarantees that the process configuration after the execution of $\tran{\ri}$
contains table entries simulating the conclusions that the execution of $\ri$
inserted into the corresponding MSR state.
The second condition helps us prove that
if $\K(m)$ holds for a given timepoint, $\Patt{m}$ also holds at the same timepoint.

\subsection{Weak Simulation}
\label{sec:weak-simulation}

Let $\rules$ be a set of MSRs.
We denote the set of weak traces with $\K$ actions removed as $\wtracesmsr_{\neq \K}(\rules)$ and the
set of weak traces with $\msgC$ actions removed as $\wtracespv_{\neq \msgC}(\tran{\rules})$.

\begin{theorem}[Weak Trace Inclusion]
	\label{th:wti}
	Let $\rules$ be a set of MSRs.
	Then
	\begin{align*}
		\lex(\wtracesmsr_{\neq \K}(\rules)) \subseteq \wtracespv_{\neq \msgC}(\tran{\rules})\,.
	\end{align*}
\end{theorem}

\begin{proof}
	Let $\vec{S} \in \execmsr(\rules)$ be an execution of length $m$ with weak trace $wtr_{msr}$.
	We show that there exists an execution $\vec{\pvconf} \in \execpv(\tran{\rules})$ of length $n$ with weak trace $wtr_{pv}$ such that
	\begin{align*}
		\mss[m] &\sim \pvconf[n] \\ 
		wtr_{pv} &= \lex(wtr_{msr})
	\end{align*}
	The proof is by induction on the length $m$ of the execution $\vec{S}$.

	\textbf{Base case} ($m=0$):
	By the definitions of the initial states $\mss[0]$ and $\pvconf[0]$,
	\begin{align*}
		\emptyset = \mss[0] &\subseteq_{tbl} \Tbl[0] = \emptyset \\	
		\emptyset = \kn(\mss[0]) &\subseteq \ded(\pvconf[0])\,,
	\end{align*}
	and thus $\mss[0] \sim \pvconf[0]$.
	Moreover, $wtr_{msr} = \epsilon = wtr_{pv}$.

	\textbf{Inductive hypothesis} ($m=k$):
	We assume that the statement holds for all executions of length $k$.

	\textbf{Inductive step} ($m=k+1$):
	\noindent Let $\vec{S}$ be an execution of length $k+1$ with $\mss[k] \withmsr{a}_{r \sigma} \mss[k+1]$ for a rule $r$ and a grounding substitution $\sigma$.
	By applying the induction hypothesis, we get an execution $\vec{\pvconf}$ of length $j$ with $\mss[k] \sim \pvconf[j]$ and $wtr'_{pv} = \lex(wtr'_{msr})$.
	We show that there exists a configuration $\pvconf[j+l]$ such that $\mss[k+1] \sim \pvconf[j+l]$ and $wtr_{pv} = \lex(wtr_{msr})$ by a case distinction on the rule $r$.

	\textbf{Case} $r \in \rules$:
	First, we show that $\tran{r \sigma}$ is applicable in $\pvconf[j]$.
	From $\mss[k] \sim \pvconf[j]$, it follows that $\mss[k][\Storage] \subseteq_{tbl} \Tbl[j]$.
	Since $r \sigma$ is applicable in $\mss[k]$, $\prems(r \sigma)_{\mid \Storage} \subseteq \mss[k]$ and hence $\prems(r \sigma)_{\mid \Storage} \subseteq_{tbl} \Tbl[j]$.

	Since $r \sigma$ is applicable in $\mss[k]$, $\prems(r \sigma)|_{\In} \subseteq \mss[k]$.
	An $\In$ fact can only be added to the state through a previous application of the \textsc{MDIn} rule.
	Since the premise of the \textsc{MDIn} rule $!\K$ is a persistent fact, it is still in $\mss[k]$.
	Thus, for all $\In(t) \in \mss[k]$ we have $t \in \kn(\mss[k])$.
	From $\mss[k] \sim \pvconf[j]$ follows $\kn(\mss[k]) \subseteq \ded(\pvconf[j])$.
	If there exists a term $t$ for $\In(t) \in \prems(r \sigma)$ for which the ProVerif attacker does not have explicit knowledge, the term can be deduced in a series of transitions starting from $\pvconf[j]$.
	These transitions have no effect on the weak trace.
	After $p < l$ transitions, we reach a configuration $\pvconf[j+p]$ with $t \in \Att$.
	By \cref{def:tran-rule}, we know that each $\In(t)$ fact is translated to an $\Pin(c, t)$ process for the public channel $c$.
	The ProVerif attacker can thus send all the terms over $c$ where they can be received by the corresponding process.
	Hence, $\tran{r \sigma}$ is applicable in $\pvconf[j]$.

	Now, we show how to get from $\pvconf[j+p]$ to $\pvconf[j+l]$, the configuration after the application of $\tran{r \sigma}$.
	The initial process is always a replication.
	With $\pvconf = \pvconf[j+p]$ and $\pvconf' = \pvconf[j+p+1]$ we have
	\begin{align*}
		\pvconf = \pvconfquad \withpv{}_{\textsc{Repl}} \; \pvconf' = (\E, \Pro \cup \{\tran{r \sigma }\}, \Tbl, \Att)\,,
	\end{align*}
	which leaves the trace unchanged.
	Next, we execute the process $\tran{r \sigma}$ step by step following \cref{def:tran-rule} closely without interleaving steps from other processes.
	From the two requirements of applicability, all required table entries exist and terms are known to the attacker (line 1 in \cref{def:tran-rule}).
	In the next step, fresh names are generated for each fresh fact in the premise of $r$ (line 2).
	We continue with the emission of action facts (line 3) for each action in $a$ in lexicographical order.
	If the action contains public variables, they are received from the ProVerif attacker.
	Similarly, storage facts from the conclusion are added to the table (line 4).
	Finally, the output facts are emulated by the $\Pout$ process (line 5).

	From the above, only the events are observable in the trace.
	Since the events are ordered lexicographically, the weak trace of the process is a lexicographical ordering of the weak trace of the MSR: $wtr_{pv} = \lex(wtr_{msr})$.

	We assume that after the execution of $\tran{r \sigma}$, we are in the configuration $\pvconf[j+l]$.
	It remains to show that $\mss[k+1] \sim \pvconf[j+l]$.
	From \cref{def:tran-rule} we know that for each storage fact in the conclusion of $r \sigma$, there is a corresponding $\Pinsert$ in $\tran{r \sigma}$, i.e., $\concls(r \sigma)_{\mid \Storage} =_{tbl} \inserts(r \sigma)$.
	Thus, at the end of the execution of $\tran{r \sigma}$, we have $\Tbl[j+l] = \Tbl[j] \cup \inserts(r \sigma)$.

	From $\mss[k] \sim \pvconf[j]$ we know that $\mss[k][\Storage] \subseteq_{tbl} \Tbl[j]$.
	The transition from $\mss[k]$ to $\mss[k+1]$ adds the storage facts in $\concls(r \sigma)$ to the state.
	The table entries that correspond to these facts are exactly $\inserts(r \sigma)$.
	From $\Tbl[j+l] = \Tbl[j] \cup \inserts(r \sigma)$ follows that $\pvconf[j+l]$ contains a corresponding table entry for each storage fact in $\mss[k+1]$, i.e., $\mss[k + 1][\Storage] \subseteq_{tbl} \Tbl[j+l]$.

	Since no new $!\K$ facts are added to the state, we have
	\begin{equation*}
		\kn(\mss[k+1]) = \kn(\mss[k]).
	\end{equation*}
	Hence, $\mss[k + 1] \sim \pvconf[j+l]$.

	\textbf{Case} $r = \textsc{MDIn}$:
	\textsc{MDIn} is the only rule that adds $\K$ actions to the trace.
	Since we filter out the $\K$ actions when comparing traces, they are not observable.
	In ProVerif, we simulate this rule by staying in $\pvconf[j]$.
	Hence, in both cases the observable behavior is the same with respect to our weak-trace comparison.
	The application of \textsc{MDIn} does not add storage or $!\K$ facts to $\mss[k+1]$.
	Therefore, $\mss[k+1][\Storage] = \mss[k][\Storage]$ and $\kn(\mss[k+1]) = \kn(\mss[k])$.
	Consequently, $\mss[k+1] \sim \pvconf[j]$.

	\textbf{Case} $r = \textsc{MDOut}$:
	The application of \textsc{MDOut} adds a $!\K$ fact to $\mss[k+1]$.
	Hence, $\kn(\mss[k+1]) = \kn(\mss[k]) \cup \{m\}$.
	In ProVerif, we simulate this rule by staying in $\pvconf[j]$.
	Since $\textsc{MDOut}$ is applicable in $\mss[k]$, we know that $\Out(m) \in \mss[k]$.
	Hence, there exist a rule $r^* \in \rules$ and a grounding substitution $\sigma^*$ such that $r^* \sigma^*$ added $\Out(m)$ to the state at a previous timepoint.
	Therefore, $\tran{r^* \sigma^*}$ has also been executed earlier including the corresponding $\Pout$ process.
	By the semantics of ProVerif's $\textsc{Out}$ transition, we know that $m$ was added to the knowledge of the attacker and thus $m \in \ded(\pvconf[j])$.
	The application of \textsc{MDOut} does not add storage facts to $\mss[k+1]$.
	Therefore, $\mss[k+1][\Storage] = \mss[k][\Storage]$.
	Consequently, $\mss[k+1] \sim \pvconf[j]$ holds.

	\textbf{Case} $r = \textsc{MDFresh}$:
	\textsc{MDFresh} allows the attacker to learn a previously generated fresh name through the rule $\Fresh$.
	Its application adds a $!\K$ fact to $\mss[k+1]$.
	Hence, $\kn(\mss[k+1]) = \kn(\mss[k]) \cup \{m\}$.
	In ProVerif, we simulate this behavior with a single $\textsc{New}$ transition adding the name to the knowledge of the attacker.
	There is no observable behavior in either case.
	The application of \textsc{MDFresh} does not add storage facts to $\mss[k+1]$.
	Therefore, $\mss[k+1][\Storage] = \mss[k][\Storage]$.
	Consequently, $\mss[k+1] \sim \pvconf[j+1]$ holds.

	\textbf{Case} $r = \textsc{MDApp}$:
	\textsc{MDApp} allows the attacker to construct a message from its knowledge.
	Its application adds a $!\K$ fact to $\mss[k+1]$.
	Hence, $\kn(\mss[k+1]) = \kn(\mss[k]) \cup \{f(m_1, \ldots, m_p)\}$ for a function symbol $f$ and $!\K(m_i) \in \mss[k]$.
	As $\mss[k] \sim \pvconf[j]$, we know that $m_i \in \ded(\pvconf[j])$.
	Thus, the attacker can deduce them in a finite number $q$ of steps and add them to its knowledge.
	In ProVerif, we simulate this behavior with a $\textsc{App}$ transition adding the term $f(m_1, \ldots, m_p)$ to its knowledge.
	There is no observable behavior in either case.
	The application of \textsc{MDApp} does not add storage facts to $\mss[k+1]$.
	Therefore, $\mss[k+1][\Storage] = \mss[k][\Storage]$.
	Consequently, $\mss[k+1] \sim \pvconf[j+q]$ holds.

	\textbf{Case} $r = \textsc{Fresh}$:
	\textsc{Fresh} allows for generating fresh names by adding a $\Fr(n)$ fact to the state.
	At this time, we do not know how the generated name will be used.
	It can either be used by another protocol rule that requires a fresh fact in its premise or by \textsc{MDFresh} to allow the attacker to generate a fresh name.
	In ProVerif, the former case is simulated by a $\textsc{Restr}$ transition and the latter by a $\textsc{New}$ transition.
	Hence, we stay in $\pvconf[j]$.
	There is no observable behavior in either case.
	The application of \textsc{Fresh} does not add storage or $!\K$ facts to $\mss[k+1]$.
	Therefore, $\mss[k+1][\Storage] = \mss[k][\Storage]$.
	Consequently, $\mss[k+1] \sim \pvconf[j]$ holds.

	\textbf{Case} $r = \textsc{MDPub}$:
	\textsc{MDPub} allows the attacker to learn public names.
	Its application adds a $!\K$ fact to $\mss[k+1]$.
	Hence, $\kn(\mss[k+1]) = \kn(\mss[k]) \cup \{m\}$.
	In ProVerif, the attacker knows all public names from the start and thus $m \in \ded(\pvconf[j])$.
	Hence, we stay in $\pvconf[j]$.
	There is no observable behavior in either case.
	The application of \textsc{MDPub} does not add storage facts to $\mss[k+1]$.
	Therefore, $\mss[k+1][\Storage] = \mss[k][\Storage]$.
	Consequently, $\mss[k+1] \sim \pvconf[j]$ holds.

	In the last two cases (\textsc{Fresh} and \textsc{MDPub}), the execution $\vec{S}$ can stutter since both rules have empty premises and can thus be applied infinitely often.
	Since the rules have no observable actions, the weak trace remains unchanged.
	The same holds for the ProVerif execution, as we stay in the same configuration for both rules.
	Hence, we can ignore these two cases of possible stuttering of the MSR LTS as weak simulation only requires the observable actions to be matched.
\end{proof}

We proved weak simulation between Tamarin and ProVerif without considering $\K$ and $\msgC$ facts, respectively.
Recall that this is not problematic as $\msgC$ facts are not part of SA-QBF and $\K$ facts are translated to $\PattC$ facts, which do not appear in the trace.

\begin{remark}
	\label{rem:no-reverse-inclusion}
	The reverse inclusion does not hold in general.
	A ProVerif execution may stop between the events of a single translated rule, producing a strict prefix of the lexicographic image of a Tamarin action multiset, and, without the instrumentation of \cref{sec:bridging-gaps}, table entries may be consumed more often than the corresponding linear facts.
	The encodings of \cref{sec:bridging-gaps} mitigate the latter in the implementation but lie outside the formal statement.
\end{remark}

\begin{remark}[Instrumented translation and formula rewriting]
	\label{rem:instrumented}
	The implementation additionally instruments events with rule-instance identifiers (\cref{sec:bridging-gaps}): each translated rule instance draws a fresh identifier and passes it as an additional argument to all of its events.
	The construction above extends verbatim, as the identifier is one further fresh name and the extra event argument does not affect the simulation.
	To exclude artificial mid-action prefixes for a cross-fact same-action conclusion, the translation may guard it with an internal $\fact{RuleCompleted}$ event that is emitted after the serialized source action and removed by projection to source events.
	Over such instrumented traces, the formula rewriting of \cref{sec:bridging-gaps} preserves satisfaction:
	\begin{enumerate*}[label=(\arabic*)]
		\item one-point elimination of definite timepoint equalities is a logical equivalence;
		\item splitting a shared time variable into per-event time variables tied by a common identifier variable is exact, since two action facts occur at the same Tamarin timepoint if and only if they belong to the same rule instance; and
		\item translating a surviving timepoint equality (resp.\ disequality) between distinct events as an equality (resp.\ disequality) of the corresponding identifier variables is exact for the same reason.
	\end{enumerate*}
	The rewritten formulas contain exactly one trace atom per time variable and no timepoint equalities between distinct atoms, as required by \cref{lem:val-translation}.
\end{remark}

With this construction, we can now show that satisfaction of closed formulas is preserved from MSR traces to ProVerif traces.

\begin{corollary}[MSR to ProVerif Satisfaction]
	\label{cor:msr-to-pv-satisfaction}
	Let $\rules$ be a set of MSRs and $t \in \tracesmsr(\rules)$ a trace.
	Let $\vec{\pvconf}$ be the ProVerif execution obtained after the construction in \cref{th:wti}.
	If there exists a valuation $\theta$ such that $(t, \theta) \models_{\ET} \varphi$ for a closed, simultaneity-free formula $\varphi$ in negation normal form whose $\K$ atoms occur only positively---this covers $\K$-free formulas in TSA-QBF as well as negation normal forms of $\neg\psi$ for $\psi \in \text{T-UKF}$---%
	then there exists a valuation $\theta'$ such that $(\vec{\pvconf}, \theta') \models_{\ET} \tran{ \varphi }$.
\end{corollary}
\begin{proof}
	Let $t'$ be the trace obtained from projecting $\vec{\pvconf}$.
	We note that $t$ and $t'$ differ only in the presence of $\K$ facts in $t$, the $\msgC$ labels in $t'$ (which do not affect SA-QBF satisfaction by \cref{lem:pvmesstr}), the number of empty steps between action facts, and the lexicographic serialization of simultaneous action facts (handled by \cref{lem:val-translation}).
	By \cref{lem:msrwtrtr,lem:pvwtrtr}, the empty steps do not influence the satisfaction of the formula and can be ignored.
  Hence, the satisfaction only depends on finding appropriate valuations for the time variables.

	Let $\wtr(t)$ and $\wtr(t')$ be the weak traces obtained from $t$ and $t'$, respectively.
	If $t$ does not contain any $\K$ facts, then $\wtr(t') = \lex(\wtr(t))$, and thus we get a satisfying valuation $\theta^*$ by \cref{lem:val-translation} with $(\wtr(t'), \theta^*) \models_{\ET} \varphi$.
	With \cref{lem:pvwtrtr}, we get a valuation $\theta'$ such that $(t', \theta') \models_{\ET} \tran{\varphi}$.

	Let us now assume $t$ and $\varphi$ contain $\K$ facts.
	We first show that if
	\begin{equation*}
		(t, \theta) \models_\ET \K(m)@i,
	\end{equation*}
	then there exists a valuation $\theta'$ such that
	\begin{equation*}
		(\vec{\pvconf}, \theta') \models_\ET \Patt{m}@i
	\end{equation*}

	The only rule that produces $\K(m)$ is \textsc{MDIn}.
	Moreover, we know that $(\K(m)@i)\theta$ is in $t$ at index $\theta(i)$.
	Let $\vec{S}$ be the MSR execution that produces $t$.
	Then there is a corresponding transition from the MSR state $\mss[\theta(i)-1]$ to $\mss[\theta(i)]$ with the rule \textsc{MDIn}.
	By construction, $\vec{\pvconf}$ contains a configuration $\pvconf[j-1]$ such that $\mss[\theta(i)-1] \sim \pvconf[j-1]$.
	Therefore, $\kn(\mss[\theta(i)-1]) \subseteq \ded(\pvconf[j-1])$.
	Since \textsc{MDIn} is applicable in $\mss[\theta(i)-1]$, it includes $!\K(m)$, i.e., $m \in \kn(\mss[\theta(i)-1])$.
	Consequently, $m \in \ded(\pvconf[j-1])$.
	Thus, $\Patt{m}$ holds at the next timepoint with $\theta'(i) = j$.

	To reconstruct the valuation for the other time variables, we proceed similarly.
	If $(t, \theta) \models_\ET \fact{F}@i$, then a rule $r$ is applicable in $\mss[\theta(i)-1]$ with $\fact{F} \in \acts(r)$.
	By construction, $\vec{\pvconf}$ contains a configuration $\pvconf[j-1]$ such that $\mss[\theta(i)-1] \sim \pvconf[j-1]$ and such that $\tran{r}$ is applicable.
	The next $\card{ \acts(r) }$ transitions produce the corresponding action facts in lexicographical order in $t'$.
	Let $n$ be the position of $\fact{F}$ in $\lex(\acts(r))$.
	Then we set $\theta'(i) = j + n - 1$.
	This construction preserves the relative order of the timepoints.
	Hence, we get a valuation $\theta'$ such that $(\vec{\pvconf}, \theta') \models_{\ET} \tran{\varphi}$.
\end{proof}

Next, we combine our results to show soundness and completeness of our translation.

\completeness*
\begin{proof}
	Fix $tr \in \tracesmsr(\rules)$ and a valuation $\theta$ with $(tr, \theta) \sattam \varphi$.
	Since $\varphi$ contains no $\K$ facts and is simultaneity-free, we have $\tran{\varphi} = \varphi$.
	By \cref{cor:msr-to-pv-satisfaction}, there are a ProVerif execution $\vec{\pvconf} \in \execpv(\tran{\rules})$, obtained via the construction of \cref{th:wti}, and a valuation $\theta'$ with $(\vec{\pvconf}, \theta') \satpv \varphi$.
	Hence $\tracespv(\tran{\rules}) \prescript{\exists}{}{\satpv} \varphi$.
\end{proof}

Due to semantic differences between attacker knowledge in Tamarin and ProVerif,
soundness is guaranteed for T-UKF components (\cref{def:tukf}), whose counterexample formulas contain $\K$ atoms only positively.
The faithful translation fragment is their satisfaction-preserving conjunctive rewrite closure (\cref{def:faithful-fragment}); the model and equational-theory conditions are enforced by the theorem hypotheses.
We prove the result componentwise and lift it through conjunction.
ProVerif's Horn clause resolution is sound but incomplete, so properties proven valid by the verification procedure are semantically valid, though the converse does not hold.

\soundness*
\begin{proof} 
	By \cref{def:faithful-fragment}, $\varphi$ rewrites to a finite conjunction of T-UKF components.
	Since the rewritings preserve satisfaction and $\tran{\varphi}$ is the conjunction of the translated components, it suffices to establish the claim for an arbitrary component, again denoted by $\varphi$.
	We prove the contraposition of the above statement:
	\begin{align*}
		\tracesmsr(\rules) \prescript{\exists}{}{\sattam} \neg\varphi \Rightarrow \tracespv(\tran{\rules}) \prescript{\exists}{}{\satpv} \tran{\neg \varphi}
	\end{align*}
	Hence, we fix $tr \in \tracesmsr(\rules)$ and a valuation $\theta$ such that
	\begin{align*}
		(tr, \theta) \sattam \neg \varphi\,.
	\end{align*}
	Let $\vec{\pvconf} \in \execpv(\tran{\rules})$ be the corresponding ProVerif execution constructed following~\cref{th:wti}.
	Let $tr'$ be the projection of the labels of $\vec{\pvconf}$.

	Since $\varphi \in \text{T-UKF}$, the negation normal form of $\neg \varphi$ contains $\K$ atoms only positively, and it is simultaneity-free because $\varphi$ is in the faithful translation fragment (\cref{def:faithful-fragment}); hence \cref{cor:msr-to-pv-satisfaction} applies, and we obtain the existence of a valuation $\theta'$ such that $(tr', \theta') \satpv \tran{\neg \varphi}$.
	As $\tran{\cdot}$ replaces $\K$ atoms by $\PattC$ atoms, $\tran{\neg\varphi} = \neg\tran{\varphi}$, which concludes the proof.

	%
	%
	%
	%
	%
	%
	%
	%
	%
	%
	%
	%
	%
	%
	%
	%
	%
	%
	%
	%
	%
	%
	%
	%
	%

	%
	%
	%
	%
	%
	%
	%
	%
  %
	%
	%
	%
	%
	%
	%
	%
	%
	%
	%
	%
	%
	%
	%
	%
	%
	%
	%
	%
	%
	%
	%
	%
	%

\end{proof}

\section{Semantic Difference of Attacker Knowledge}

\begin{example}
	\label{ex:k-vs-attacker}
	Consider the MSR
	\begin{equation*}
	  A \colon \msr{\In(\tlst{a, b})}{\fact{A}(a)}{}
	\end{equation*}
	and the translated process $!A$ with
	\begin{align*}
		\Plet \; A {}={} &\Pin(c, x);                    \\
		                 & \Plet\,\,a\,\,=\,\,g_{a,\langle a,b\rangle}(x)\ \Pin \\
		                 & \Plet\,\,b\,\,=\,\,g_{b,\langle a,b\rangle}(x)\ \Pin \\
		                 & \Pevent\,\,\fact{A}(a).
	\end{align*}
	The trace formula
	\begin{equation*}
		\forall\ i, \; m. \; \fact{A}(m)@i \implies \exists \; j.\;\K(m)@j
	\end{equation*}
	does not hold in Tamarin.
	The counterexample is the trace
	\begin{equation*}
		(\K(\tlst{a,b}); \fact{A}(a)).
	\end{equation*}
	But ProVerif can prove the formula (with $\PattC$ for $\K$), as for $\fact{A}(a)$ to appear, $\tlst{a,b}$ must be deducible, and thus $a$ must be deducible, too.
\end{example}

\end{full}

\end{document}